\documentclass[11pt,letterpaper]{article}

\usepackage{geometry}
\usepackage{natbib}

\usepackage{titlesec}

\usepackage{microtype}
\usepackage{subfigure}
\usepackage{booktabs}
\usepackage{amsmath}
\usepackage{amsfonts}
\usepackage{amssymb}
\usepackage{amsthm}
\usepackage{amsbsy}
\usepackage{graphicx}
\usepackage{mathtools}
\usepackage{IEEEtrantools}
\usepackage{verbatim}
\usepackage{url}
\usepackage{bbm}
\usepackage{multirow}
\usepackage{enumitem}
\usepackage{xspace}
\usepackage{footnote}
\usepackage{hyperref}
\usepackage[svgnames]{xcolor}
\usepackage[capitalise,nameinlink]{cleveref}
\hypersetup{colorlinks={true},linkcolor={DarkBlue},citecolor=[named]{DarkGreen},urlcolor=black}

\usepackage{tikz}
\newcommand*\circled[1]{\tikz[baseline=(char.base)]{
    \node[shape=circle,draw,inner sep=2pt] (char) {#1};}}
\newcommand*\dotcircled[1]{\tikz[baseline=(char.base)]{
    \node[shape=circle,draw,dotted,inner sep=2pt] (char) {#1};}}
\newcommand*\dashcircled[1]{\tikz[baseline=(char.base)]{
    \node[shape=circle,draw,dashed,inner sep=2pt] (char) {#1};}}
\usepackage{algorithm}
\usepackage[indLines=true]{algpseudocodex}
\renewcommand{\algorithmiccomment}[1]{\bgroup\hfill\small\textcolor{gray}{//~#1}\egroup}
\algnotext{EndFor}
\algnotext{EndIf}
\algnotext{EndWhile}
\algnotext{EndProcedure}
\usepackage{bm}

\usepackage{titling}
\usepackage[noblocks]{authblk}

\theoremstyle{plain}
\newtheorem{theorem}{Theorem}[section]
\newtheorem{proposition}[theorem]{Proposition}
\newtheorem{lemma}[theorem]{Lemma}
\newtheorem{claim}[theorem]{Claim}
\newtheorem{corollary}[theorem]{Corollary}
\theoremstyle{definition}
\newtheorem{definition}[theorem]{Definition}

\newtheorem{remark}[theorem]{Remark}

\usepackage[tt=false]{libertine}

\usepackage{soul}
\setuldepth{d}

\makeatletter
\newcommand*{\whiten}[1]{\llap{\textcolor{white}{{\the\SOUL@token}}\hspace{#1pt}}}
\DeclareRobustCommand*\myul{%
    \def\SOUL@everyspace{\underline{\space}\kern\z@}%
    \def\SOUL@everytoken{%
     \setbox0=\hbox{\the\SOUL@token}%
     \ifdim\dp0>\z@
        \raisebox{\dp0}{\underline{\phantom{\the\SOUL@token}}}%
        \whiten{1}\whiten{0}%
        \whiten{-1}\whiten{-2}%
        \llap{\the\SOUL@token}%
     \else
        \underline{\the\SOUL@token}%
     \fi}%
\SOUL@}
\makeatother

\newcommand{\bmu}{\ensuremath{\boldsymbol{\mu}}\xspace}
\newcommand{\mms}{\ensuremath{\textrm{\MakeUppercase{mms}}}\xspace}

\newcommand{\prop}{\ensuremath{\textrm{\MakeUppercase{prop}}}\xspace}
\newcommand{\efo}{\ensuremath{\textrm{\MakeUppercase{ef}1}}\xspace}
\newcommand{\eft}{\ensuremath{\textrm{\MakeUppercase{ef}2}}\xspace}
\newcommand{\efx}{\ensuremath{\textrm{\MakeUppercase{efx}}}\xspace}
\newcommand{\ef}{\ensuremath{\textrm{\MakeUppercase{ef}}}\xspace}

\DeclareMathOperator*{\argmin}{arg\,min}

\allowdisplaybreaks

\usepackage[skip=.36\baselineskip plus 2pt]{parskip}
\usepackage{etoolbox}

\makeatletter

\patchcmd\deferred@thm@head
  {\addvspace{-\parskip}}
  {}
  {}{\typeout{\string\deferred@thm@head patch failed!}}

\makeatletter

\begin{document}

\author[1,2]{Georgios Amanatidis}
\affil[1]{Department of Informatics, Athens University of Economics and Business, Athens, Greece.}
\affil[2]{Archimedes, Athena Research Center, Athens, Greece.}
\author[1,2]{Alexandros Lolos}
\author[1,2,3]{Evangelos Markakis}
\affil[3]{Input Output Global (IOG), Athens, Greece. }
\author[4]{Victor Turmel}
\affil[4]{Institut de Mathématique, Université Paris-Saclay, Orsay, France. }

\title{Online Fair Division for Personalized $2$-Value Instances}

\predate{}
\postdate{}
\date{}

\maketitle

\begin{abstract}
\noindent We study an online fair division setting, where goods arrive one at a time and there is a fixed set of $n$ agents, each of whom has an additive valuation function over the goods. Once a good appears, the value each agent has for it is revealed and it must be allocated immediately and irrevocably to one of the agents. It is known that without any assumptions about the values being severely restricted or coming from a distribution, very strong impossibility results hold in this setting \citep{HePPZ19,ZhouBW23}. To bypass the latter, we turn our attention to instances where the valuation functions are restricted. In particular, we study \emph{personalized $2$-value instances}, where there are only two possible values each agent may have for each good, possibly different across agents, and we show how to obtain worst case guarantees with respect to well-known fairness notions, such as \emph{maximin share fairness} and \emph{envy-freeness up to one (or two) good(s)}. We suggest a deterministic algorithm that maintains a $1/(2n-1)$-\mms allocation at every time step and show that this is the best possible any deterministic algorithm can achieve if one cares about \emph{every single} time step; nevertheless, eventually the allocation constructed by our algorithm becomes a $1/4$-\mms allocation. To achieve this, the algorithm implicitly maintains a fragile system of priority levels for all agents. Further, we show that, by allowing some limited access to future information, it is possible to have stronger results with less involved approaches. In particular, by knowing the values of goods for $n-1$ time steps into the future, we design a matching-based algorithm that achieves an \efo allocation every $n$ time steps, while always maintaining an \eft allocation. Finally, we show that our results allow us to get the first nontrivial guarantees for additive instances in which the ratio of the maximum over the minimum value an agent has for a good is bounded.   
\end{abstract}
\thispagestyle{empty}

\newpage

\setcounter{page}{1}

%%%%%%%%%%%%%%%%%%%%%%%%%%%%%%%%%%%%%%%%%
%%%%%%%%%%%%%%%%%%%%%%%%%%%%%%%%%%%%%%%%%
\section{Introduction}
\label{sec:intro}
%%%%%%%%%%%%%%%%%%%%%%%%%%%%%%%%%%%%%%%%%
%%%%%%%%%%%%%%%%%%%%%%%%%%%%%%%%%%%%%%%%%
The problem of fairly allocating a set of resources to a set of agents without monetary exchanges was mathematically formalized only in the 1940's by \cite{Steinhaus49}---along with his students Banach and Knaster. Nevertheless, fair division is such a fundamental concept  that the celebrated \emph{Cut \& Choose} protocol already appears in ancient literature. 
For this reason, since the formal introduction of the problem,  numerous variants, each under different assumptions and constraints, have been studied in mathematics, economics, political science and, more recently, in computer science. 
Specifically, as far as the latter is concerned, the algorithmic nature of most questions one may ask about fair division problems has lead to a flourishing literature on the topic, often on variants that deal with discrete, indivisible resources. 

In the most standard variants of the problem a set of resources and a set of agents, each equipped with a valuation function, are given as an input and one would like to produce a, complete or partial, partition of the resources to the agents, so that some predetermined \textit{fairness criterion} is satisfied. Here we study a version of the problem where the set of $n$ agents is indeed given, but the indivisible goods arrive in an \textit{online} fashion. That is, at each time step a new good appears, its value for the agents is revealed, and the good must be irrevocably assigned to a single agent before the next good arrives. Online fair division problems of similar nature appear in many real-world scenarios, like, e.g., the operation of food banks, donation distribution in disaster relief, limited resource allocation within an organization like a hospital or a university, or memory and CPU management in cloud computing. 

Despite their wide applicability, however, online fair division settings have not been studied nearly as much as their offline counterparts. And admittedly, 
%Besides the fact that moving to more complex settings usually comes after a topic has matured sufficiently, 
there is a good reason for the relative scarcity of such works. Namely, it is relatively easy to show strong impossibility worst-case results, even if one aims for rather modest fairness guarantees. For example, when the agents have additive valuation functions and the items are goods (i.e., they do not have a negative value for any agent) that arrive adversarially---as is the case here---there is no deterministic algorithm that achieves \textit{any} positive approximation factor with respect to \textit{maximin share fairness} (\mms) \citep{ZhouBW23} or to \textit{envy-freeness up to one good} (\efo) \citep{HePPZ19}. Even if one allows some item reallocations in order to fix that, the number of reallocations cannot be bounded by any function of $n$ \citep{HePPZ19}. 

Of course, there are ways to bypass these results, like assuming that the item values are drawn from distributions \citep{BenadeKPP18,zeng2020dynamic}, restricting the valuation functions \citep{AleksandrovAGW15} or the number of distinct agent types \citep{KulkarniMS25}, allowing items to be reallocated \citep{HePPZ19}, relaxing the requirement for fairness guarantees at the intermediate time steps \citep{CooksonES25}, augmenting the input with predictions \citep{BanerjeeGGJ22} or even with full knowledge of the whole instance \citep{ElkindLLNT24,CooksonES25}. 
Here the main approach we take is to focus on designing \textit{deterministic} algorithms for instances with \textit{restricted} valuation functions, although we also explore how information about the future allows for stronger results with simpler algorithms. 
In particular, we mostly consider instances where each agent $i$ has two values for the goods: a low value $\beta_i$ and a high value $\alpha_i$. 
These instances, which we call \textit{personalized $2$-value instances}, generalize the binary case and the setting where there are only two types of goods; until now, these were the only cases where positive worst-case results were known (by \cite{AleksandrovAGW15} and \cite{ElkindLLNT24}, respectively). Consequently, personalized $2$-value instances capture the natural dichotomy between highly desirable and less desirable goods in a more nuanced way which also varies across agents. 
Moreover, there is a natural way of approximating any additive instance via a $2$-value instance by just setting a threshold for each agent $i$ and rounding everything up (to the ``high value'' $\alpha_i$) or down (to the ``low value'' $\beta_i$) accordingly. Although this approximation is not always meaningful, it does give nontrivial guarantees for agents whose values are bounded by lower and upper bounds that do not differ by too much, as we argue in Section \ref{sec:beyond_2value}. 

It should be noted that this type of restriction has drawn significant attention in the  fair division literature recently. 
The study of \textit{$k$-value instances} (often called \textit{$k$-valued} or, in the case where $k=2$, \textit{bi-valued}), where there are at most $k$ possible common values that can represent the value an agent has for a good, was introduced by \citet{ABFHV21} as a way to study the existence of \efx allocations under a restriction that made the question easier, yet not straightforward.
Indeed, it turns out that, in many cases, $k$-value instances seem to strike a good balance between maintaining the flavor and many of the challenges of the general additive version of the problem and allowing nontrivial positive results, even for $k = 2$ or $3$.
Consequently, there is a recent line of work studying fair division questions under such restrictions (e.g., \citep{AkramiC0MSSVVW22,GargMQ22,AzizLRS23,AFS24,FitzsimmonsVZ24}), often allowing for the $k$ values to be different per agent (hence, \emph{personalized}), first studied by \citet{MurhekarG21} for general $k$, under the name \textit{$k$-ary instances}. 

Finally, as we mentioned above, we also explore another relaxation of the problem which is somewhat orthogonal to that of restricting the valuation functions, and aims to alleviate the lack of information due to the online arrival of the goods. Several recent works augment their online algorithms  with additional information about the future in the form of (possibly erroneous) predictions \citep{lykouris2021competitive,BanerjeeGGJ22,BanerjeeGHJM023,BalkanskiGTZ24,BenomarP24}. Here we follow an approach closer to \citet{HePPZ19}, \citet{CooksonES25}, and \citet{ElkindLLNT24}, who assume that the whole instance can be known in advance and what truly remains online is the allocation process itself. As the number of goods could be much larger than the number of agents in many scenarios, we feel that knowing the whole instance is a rather strong assumption; instead, we allow some of our algorithms to see into the future only for a number of time steps that is comparable to the number of agents. In our setting this turns out to be enough for getting simpler algorithms with solid guarantees.  
\medskip

\noindent{\textbf{Contribution and Technical Considerations.}}
Following the aforementioned line of work on restricting the valuation functions, we explore what is possible for (personalized) $2$-value instances in deterministic online fair division. We obtain positive and negative results, which we then extend beyond our main setting. Specifically:
\begin{itemize}[leftmargin=20pt,topsep=3pt,itemsep=3pt]
    \item The impossibility results one has for general additive instances persist here as well, albeit milder. We show that, for any $\varepsilon > 0$, no algorithm can guarantee $(1/2 + \varepsilon)$-\efo at every time step, even for two agents (Theorem \ref{thm:imposibility_0.5EF1}), or $(1/(2n - 1) + \varepsilon)$-\mms at every time step for general $n$ (Theorem \ref{thm:imposibility_MMS}). 

    \item We present an algorithm with tight approximation guarantee with respect to \mms. In particular, our Algorithm \ref{alg:main_alg} guarantees $1/(2n - 1)$-\mms at every time step, which improves to $\Omega(1)$ at every time step, assuming that $m$ is large enough (Theorem \ref{thm:main_algorithm} and Corollary \ref{cor:type1}). 

    \item We demonstrate that even very limited knowledge of the future may help significantly, as our Algorithm \ref{alg:foresight_2} guarantees \efo at every \textit{other} time step for two agents just by looking one step ahead into the future (Theorem \ref{thm:foresight_2}). 

    \item More generally, we show that having a foresight of $n-1$ goods into the future suffices in order to guarantee \eft at every  time step and \efo at every $n$ time steps for $n$ agents (Algorithm \ref{alg:foresight_n} and Theorem \ref{thm:foresight_n-b}). The latter also implies a $1/n$-\mms approximation at every $n$ time steps, whereas achieving $(1/n + \varepsilon)$-\mms at every time step is impossible, even if  the whole instance is fully known in advance.

    \item We provide a simple reduction that allows us to translate our results to general additive instances at the expense of a multiplicative factor that depends on the largest ratio between the values of any two goods. To the best of our knowledge, these are the first positive results in this setting (Theorem \ref{thm:reduction} and Corollaries \ref{cor:interval-main} and \ref{cor:interval-foresight}).
\end{itemize}
While the main ideas in all of our results are very intuitive at a high level, turning them into tight or best-possible statements is far from trivial. 
Our impossibility results (especially Theorem \ref{thm:imposibility_MMS}) cannot rely on boosting values as needed (as it is commonly done in the literature), given the nature of the valuation functions we study. Instead, we need a family of instances fully adjusted to what \textit{any} algorithm could do in order to determine the allocation with enough precision to get what turns out to be a tight bound. 
Algorithm \ref{alg:main_alg}, despite its several cases and careful book-keeping, is based on the simple idea that throughout the allocation process an agent should gain higher priority the more value they lose to others. Although our parametrization of the indices that imply this priority (one out of every $2n-1$ goods in general; one out of $3n-2$ high-valued goods) may seem rather loose, not only is it tight but it requires an involved analysis that includes distinct inductive proofs (that differ drastically) for the initial and later phases of the algorithm, respectively. 
Algorithms \ref{alg:foresight_2} and \ref{alg:foresight_n} are simpler, especially the former, but there are subtleties here as well. Ideally, what we would like to achieve between time steps $kn$ and $(k+1)n$ would be to get an \efo allocation by taking the union of two \efo allocations. However, in general, this fails trivially, even for agents with binary valuation functions. Instead, we take advantage of the structure of our instances and built the second \efo allocation (the one with the current and predicted goods) so that its envy graph ``cancels out'' the problematic edges of the envy graph of the current allocation, via carefully defined auxiliary valuation functions.
\medskip

\noindent{\textbf{Further Related Work.}}
There is a vast literature on fair division, both with divisible and with indivisible resources. For a recent survey on the latter, see \citet{AmanatidisABFLMVW23}. Here we focus on online fair division settings, mostly with indivisible items. 

\citet{AleksandrovAGW15} introduced the setting we study here. Their work focused more on binary instances and on welfare guarantees. Subsequent works studied the mechanism design version of the problem mostly for binary instances \citep{AleksandrovW17} and explored the limitations of achieving fairness notions like \efo with general additive or even with non-additive binary valuations functions \citep{AleksandrovW19}. The results of these early works are summarized in the survey of \citet{AleksandrovW20}. 
A viable direction in order to bypass the existing strong negative results is to assume that there are underlying distributions with a bounded support. In such a setting \citet{BenadeKPP18} show that it is possible to achieve envy that grows sublinearly with time with a very simple algorithm and that this is asymptotically best possible, whereas \citet{zeng2020dynamic} study the  compatibility of efficiency and envy in variants of the setting.
However, the works that are closer to ours, one way or another, are those of \citet{HePPZ19}, \citet{ZhouBW23}, \citet{CooksonES25}, and \citet{ElkindLLNT24}.

\citet{HePPZ19} show that it is impossible to achieve temporal-\efo (or even any nontrivial approximation of it) unless a linear fraction of the goods are reallocated. Then they design algorithms (occasionally augmented with full knowledge of the future) that achieve temporal-\efo with a bounded number of reallocations.
\citet{ZhouBW23} focus on temporal-\mms and show that no approximation is possible for goods beyond the case of two agents. Then they turn their attention to chores (i.e., items that no agent values positively) where they present an algorithm with constant approximation guarantee. It should be noted here that although their $0.5$-temporal-\mms algorithm for $n=2$ seems to contradict our Theorem \ref{thm:imposibility_MMS}, this is made possible by the additional assumption that $v_i(M)$ is known for all $i\in N$, which we do not make here.
The very recent works of \citet{ElkindLLNT24} and \citet{CooksonES25} formalize the notion of \emph{temporal fairness} (although in slightly different ways) and assume that the full information about an instance is known upfront. However, even under this strong assumption, there still are impossibility results and, hence, both papers mostly focus on further restrictions.  \citet{CooksonES25} obtain positive results at every step of the allocation for the case of two agents, and for the overall allocation for instances where the agents agree on the  ordering over the goods or for the variant of the problem where the same set of goods arrives at every time step.  \citet{ElkindLLNT24}, on the other hand, show temporal \efo guarantees for two agents, two types of items,  generalized binary
valuations, and for unimodal preferences.

There is also a line of work where the setting is very similar to ours but the items that arrive online are divisible \citep{GkatzelisPT21,Barman0M22,BanerjeeGGJ22,BanerjeeGHJM023}. The similarities, however, are only superficial, as the fractional assignment of even a few goods allows us to bypass most strong impossibility results.
Finally, online fair division with indivisible items has been very recently studied in the context of bandit learning \citep{YamadaKAI24,ProcacciaS024}. 

Of course, besides the online setting where the items arrive over time, a different scenario is to have a known set of resources and assume that agents arrive in an online fashion. Indeed, there is a significant amount of work in this direction, with an emphasis on indivisible resources \citep{KalinowskiNW13,KashPS14,SinclairBY21,VardiPF22,BanerjeeHS24}. 
For indivisible resources, the very recent work of \citet{KulkarniMS25} obtains the first solid maximin  share fairness guarantees for this setting assuming that the agents fit into a limited number of different types. 
Nevertheless, these settings have very different flavor than ours.

%%%%%%%%%%%%%%%%%%%%%%%%%%%%%%%%%%%%%%%%%
%%%%%%%%%%%%%%%%%%%%%%%%%%%%%%%%%%%%%%%%%
\section{Preliminaries}
\label{sec:prelims}
%%%%%%%%%%%%%%%%%%%%%%%%%%%%%%%%%%%%%%%%%
%%%%%%%%%%%%%%%%%%%%%%%%%%%%%%%%%%%%%%%%%

Let $N = [n] = \{1, 2, \ldots, n\}$ be  set of agents and  $M = \{g_1, g_2, \ldots, g_m\}$ be a set of indivisible goods, where $n, m \in \mathbb{N}$. The high-level goal is to assign the goods of $M$ to the agents of $N$ in a way that is considered fair according to a well-defined fairness criterion. We usually call this assignment of goods to agents an \emph{allocation}, which can be \emph{partial} (if not all goods of $M$ have been assigned) or complete (i.e., it defines a partition of $M$). Formally, an allocation $(A_1, \ldots, A_n)$ is an ordered tuple of disjoint subsets of $M$; we often call the set $A_i$ the \emph{bundle} of agent $i$.

Each agent $i$ is associated with an \emph{additive} set function $v_i: 2^M \to \mathbb{R}_{\ge 0}$, where $v_i(S) = \sum_{g \in S}v_i(\{g\})$ represents the value of agent $i$ for the subset $S$ of goods. When $S$ is a singleton, we usually write $v_i(g)$ rather than $v_i(\{g\})$.
Of course, valuation functions of the agents can be more general but in this work we only study special cases of additive instances of the problem, i.e., instances where all agents have additive valuation functions that are restricted in some way.
In particular, we focus on \emph{personalized $2$-value} and \emph{personalized interval-restricted} instances. 

\begin{definition}[Personalized $2$-Value Instances]\label{def:2-value}
We say that an instance of the problem is a \emph{personalized $2$-value instance}, if for any $i \in N$ the function $v_i$ is additive and there are $\alpha_i \ge \beta_i \ge 0$, such that for any $g \in M$, it holds that $v_i(g) \in \{\alpha_i, \beta_i\}$. When $\alpha_i = \alpha$, $\beta_i = \beta$, for all $i\in N$, we call this a \emph{$2$-value instance}.
\end{definition}

One could reasonably claim here that the interesting case is when $\alpha_i > \beta_i > 0$ in Definition \ref{def:2-value}; indeed we say that such agents are of \textit{type 1}. Similarly, if $\alpha_i = \beta_i > 0$, agent $i$ is of \textit{type 2} and if $\alpha_i > \beta_i = 0$, agent $i$ is of \textit{type 3}. The remaining agents (i.e., $\alpha_i = \beta_i = 0$) are of \textit{type 0} and are completely irrelevant, as they are trivially satisfied (namely, they see the allocation as being temporal-\ef) by an empty bundle. So, without loss of generality, we may assume that the instances we consider only contain agents of types 1, 2, and 3. Of course, agents of type 2 and 3 are easier to satisfy (see, e.g., Corollary \ref{cor:types2-3-4}). As a final observation about agent types, we note that for the fairness notions we introduce below and study in this work, scaling the valuation functions has no effect on the fairness guarantees of any given allocation. Thus, it is also without loss of generality, to assume that $\beta_i =1$ if agent $i$ is of type 1 or 2, or that $\alpha_i = 1$ if it is of type 3.

The next definition aims to capture the continuous analog of $2$-value instances, i.e., we would like all the values agent $i$ may have for a good to be in the interval $[\beta_i, \alpha_i]$. However, as we just mentioned above, scaling the valuation functions is irrelevant for the fairness notions we consider. That is, any general additive instance can be trivially transformed to an equivalent instance where, for any $i \in N$ and $g \in M$, it holds that $v_i(g) \in [0, 1]$. Thus, in order for the restriction to be meaningful in our context, it should hold that $\beta_i > 0$ for all $i\in N$. By scaling appropriately, this is equivalent to asking that  $\beta_i =1$ for all $i\in N$.

\begin{definition}[Personalized Interval-Restricted Instances]\label{def:interval}
We say that an instance of the problem is a \emph{personalized interval-restricted instance}, if for any $i \in N$ the function $v_i$ is additive and there is $\alpha_i > 1$, such that for any $g \in M$, it holds that $v_i(g) \in [1, \alpha_i]$. When $\alpha_i = \alpha$, for all $i\in N$, we call this an \emph{interval-restricted instance}.
\end{definition}

By maintaining a bound on $\max_{i\in[n]}\sqrt{\alpha_i\,}$, i.e., the ratio of the highest over the lowest value an agent has for a good in Definition \ref{def:interval}, we can reduce the problem of dealing with personalized interval-restricted instances to dealing with personalized $2$-value proxy instances instead (albeit at the expense of the approximation ratio guarantees; see Section \ref{sec:beyond_2value}). 

In standard---offline---fair division settings, all the goods of $M$ are known and available to be used as an input to an allocation algorithm. Here we consider an \textit{online} fair division setting, where the goods arrive one at a time; the set $M$ (or even its cardinality, $m$) is not known a priori. When a good $g$ arrives, its value for each agent is revealed and it needs to be added to the bundle of some agent immediately and irrevocably. In general, we associate a distinct time step with each good and it is often (although not always) convenient to implicitly rename the goods so that $g_k$ is the $k$-th good in order of arrival and arrived at (or rather \textit{triggered}) time step $t = k$. Also, we often use $n_h(i, t) = |\{ g\in M \,|\, v_i(g) = \alpha_i \text{ and } g \text{ is one of the first $t$ goods that arrived}\}|$ to denote the number of high-valued goods among the first $t$ goods from the perspective of agent $i$.

As we mentioned in our introduction, there are no distributional assumptions about the arrival of the goods and in our results we follow a worst-case analysis. In Section \ref{sec:foresight_whole}, however, we assume that our instances can be augmented with limited information about the future. We say that an online instance is augmented with \emph{foresight of length $\ell$} if every time a good $g$ arrives (and still needs to be allocated immediately and irrevocably), we also get a preview of the $\ell$ next goods.

We mentioned above that we want to produce allocations that are \emph{fair} in some manner. We formalize this by introducing our main fairness notions, \emph{approximate EF\,$k$} and \emph{approximate \mms}.
Envy-freeness up to $k$ goods (EF$k$) is a relaxation of envy-freeness, introduced by \citet{LMMS04} and formally defined by \citet{Budish11} for $k=1$. For instance, according to \efo some envy is acceptable, as long as it can be eliminated by the hypothetical removal of a single good.

\begin{definition}[$\rho$-Envy-Freeness, $\rho$-EF$k$]
Given a partial allocation $\mathcal{A} = (A_1, A_2, \ldots, A_n)$, constants $\rho \in (0, 1]$ and $k\in \mathbb{Z}_{>0}$, and two agents $i, j \in N$, we say that 
\begin{itemize}[leftmargin=20pt,itemsep=3pt]
    \item[-] agent $i$ \emph{is $\rho$-envy-free ($\rho$-EF) towards} agent $j$, if $v_i(A_i) \geq \rho \cdot v_i(A_j)$;
    \item [-] agent $i$ \emph{is $\rho$-EF$k$ towards} agent $j$, if there is a set $S\subseteq A_j$ with $|S|\le k$, such that $v_i(A_i) \geq \rho \cdot v_i(A_j \setminus S)$.
\end{itemize}
The allocation is called $\rho$-\ef (resp.~$\alpha$-EF$k$) if every agent $i \in N$ is $\rho$-envy-free (resp.~$\rho$-EF$k$) towards any other agent $j \in N$. When $\rho = 1$, we drop the prefix and write EF$k$ rather than $1$-EF$k$. 
\end{definition}

Maximin share fairness, introduced by \cite{Budish11}, is a share-based notion, like proportionality, and can be interpreted via a thought experiment inspired by the famous cut-and-choose protocol. The idea is to give to each agent at least as much value as it could get if it partitioned the goods into disjoint sets (as many as the agents) and kept the worst among them. 

\begin{definition}[$\rho$-\prop, $\rho$-\mms]
For  a partial allocation $\mathcal{A} = (A_1, A_2, \ldots, A_n)$, let $\Pi(n, \mathcal{A})$ be the set of possible partitions of the set $S = {\bigcup_{j=1}^n}A_j$ into $n$ subsets.
Given $\mathcal{A}$ above, a constant $\rho>0$, and an agent $i\in N$, we say that 
\begin{itemize}[leftmargin=20pt,itemsep=3pt]
    \item[-] $\mathcal{A}$ \emph{is $\rho$-proportional} for agent $i$, if $v_i(A_i) \geq \rho \cdot v_i(S) / n$;
    \item [-] $\mathcal{A}$ \emph{is $\rho$-\mms} for agent $i$, if $v_i(A_i) \geq \rho \cdot \bmu_i^n(S)$, where $\bmu_i^n(S) = \!\!\displaystyle\max_{\mathcal{B} \in \Pi(n, \mathcal{A})} \min_{B_j \in \mathcal{B}} v_i(S)$ is the \emph{maximin share} of agent $i$; when $n$ is clear and $S$ is a function of time $t$, we write $\bmu_i(t)$ instead $\bmu_i^n(S)$.
\end{itemize}
The allocation is called $\rho$-\prop (resp.~$\rho$-\mms) if it is $\rho$-proportional (resp.~$\rho$-\mms) for every agent $i \in N$. When $\rho = 1$, we write \mms rather than $1$-\mms. 
\end{definition}

It is known, and very easy to derive from the definitions, that $\rho$-envy-freeness implies not only $\rho$-envy-freeness up to $k$ goods, for any $k>0$, but also $\rho$-proportionality, which itself implies $\rho$-maximin share fairness. 

As our problem is online, we do not only care about the final (complete) allocation. As a result, we will make statements of the form ``the allocation at time step $t$ is $\rho_1$-EF$k$ (resp.~$\rho_2$-\mms)'' meaning that we consider and evaluate the allocation that has been constructed up to time step $t$ as if the complete set of goods was only what we have seen so far. If each one of the partial allocations produced by an algorithm satisfies the same fairness guarantee, then one talks about \textit{temporal fairness}, as this was formalized by \citet{ElkindLLNT24} and \citet{CooksonES25}.

\begin{definition}[Temporal Fairness]
	Consider a sequence of partial allocations $\mathcal{A}^t = (A_1^t, A_2^t, \ldots, A_n^t)$, for $t\in \mathbb{Z}_{\ge0}$, such that $A_i^t\subseteq A_i^{t+1}$ for any $i\in N$ and any $t\ge0$. If $\mathcal{A}^t$ is $\rho_1$-EF$k$ (resp.~$\rho_2$-\mms) for all $t\in \mathbb{Z}_{\ge0}$, then we say that the sequence of allocations $(\mathcal{A}^t)_{t\ge0}$ is $\rho_1$-\emph{temporal}-EF$k$ (resp.~$\rho_2$-\emph{temporal}-\mms).   
\end{definition}

When referring to the allocation iteratively built by an algorithm, we may abuse the terminology and say that the algorithm computes a $\rho_1$-temporal-EF$k$ (resp.~$\rho_2$-temporal-\mms) allocation, rather than talking about a sequence of allocations.

\begin{remark}
 Suppose that we have a personalized interval-restricted instance or a personalized $2$-value instance with type 1 agents. Let  $\alpha_* = \max_{i\in[n]}\sqrt{\alpha_i\,}$.
It should be noted that here one can get a trivial $1/\alpha_*$ approximation with respect to temporal \efo or \mms. Indeed, this is done by completely ignoring the values and allocating the goods in a round-robin fashion. Although, in general, $1/\alpha_*$ can be arbitrarily worse than the approximation factors we achieve throughout this work, in the special case where $\alpha_*$ is a small constant (e.g., between $1$ and $2$), it would be preferable to follow this trivial approach instead.
\end{remark}

%%%%%%%%%%%%%%%%%%%%%%%%%%%%%%%%%%%%%%%%%
%%%%%%%%%%%%%%%%%%%%%%%%%%%%%%%%%%%%%%%%%
\section{Impossibility Results Persist Even for $2$-Value Instances} 
\label{sec:imposibilities}
%%%%%%%%%%%%%%%%%%%%%%%%%%%%%%%%%%%%%%%%%
%%%%%%%%%%%%%%%%%%%%%%%%%%%%%%%%%%%%%%%%%
The strong impossibility results in the literature \citep{HePPZ19,ZhouBW23} typically exploit the following pattern: a bad decision is made about the very first good, due to lack of information, and this propagates throughout a linear number of goods; then, right about when the value of the allocated goods starts to balance out, this ``bad'' sequence is replicated but with all values scaled up by a large factor, and this cycle is repeated as necessary. One could hope that for $2$-value instances there is not enough flexibility for creating instances that force any algorithm to perform poorly. While there is some truth in this, in the sense that things cannot go arbitrarily bad, we still get nontrivial impossibility results.

\begin{theorem}\label{thm:imposibility_0.5EF1}
Let $\varepsilon > 0$ be a constant. There is no deterministic algorithm that always builds an allocation which is $(1/2 + \varepsilon)$-temporal-\efo, even for $2$-value instances with only two agents.  
\end{theorem}

\begin{proof}
Suppose we have a deterministic  algorithm $\mathcal{A}$ for the problem. We begin with the simple observation that when the very first good has the same value for both agents, it is without loss of generality to assume that $\mathcal{A}$ assigns it to agent 1; if not, we just rename the agents in the following argument. 

So, consider a stream of goods that begins with $g_1, g_2$, such that $v_1(g_1) = v_2(g_1) = 1$ whereas $v_1(g_2) = 5$ and $v_2(g_2) = 1$. Given that $g_1$ is added to $A_1$, either $g_2$ is also added to $A_1$ and the resulting allocation $(\{g_1, g_2\}, \emptyset)$ is not $(1/2 + \varepsilon)$-\efo from the point of view of agent 2, or $g_2$ is added to $A_2$; we assume the latter. Next, consider a good $g_3$, such that $v_1(g_3) = v_2(g_3) = 1$. There are two cases here, depending on what $\mathcal{A}$ does with $g_3$, both illustrated below:

% \begin{table}[h!]
\begin{center}  
\begin{tabular}{lcccccccccccc}
 & $g_1$ & $g_2$ & $g_3$ & $g_4$  &\ldots &  & $g_1$ & $g_2$ & $g_3$ & $g_4$ & $g_5$ & \ldots \\
agent 1: \hspace{1em} & \circled{1} & 5 & \circled{1} & 5 & \ldots & or & \circled{1} & 5 & 1 & 1 & 5 &\ldots \\
agent 2: \hspace{1em} & 1 & \circled{1} & 1 & 5 & \ldots & \hspace{4em} & 1 & \circled{1} & \circled{1} & 5 & 5 &\ldots 
\end{tabular}
\end{center}
% \end{table}

\noindent\myul{\textbf{Case 1:} $g_3$ is given to agent 1.} In this case, consider a next good $g_4$, such that $v_1(g_4) = v_2(g_4) = 5$. Whoever gets $g_4$, the resulting allocation, $(\{g_1, g_3, g_4\}, \{g_2\})$ or $(\{g_1, g_3\}, \{g_2, g_4\})$, is not $(1/2 + \varepsilon)$-\efo. Indeed, $(\{g_1, g_3, g_4\}, \{g_2\})$ is not $(1/2 + \varepsilon)$-\efo from the point of view of agent 2, since 
\[1 = v_2(A_2) < (1/2 + \varepsilon) \min_{g\in A_1} (v_2(A_1) - v_2(g)) = 1 +2 \varepsilon \,.\]
Similarly, $(\{g_1, g_3\}, \{g_2, g_4\})$  is not \efo from the point of view of agent 1, since now we have
\[2 = v_1(A_1) < (1/2 + \varepsilon) \min_{g\in A_2} (v_1(A_2) - v_1(g)) = 2.5 +5 \varepsilon \,.\]
% \smallskip

\noindent\myul{\textbf{Case 2:} $g_3$ is given to agent 2.} Here we first have an intermediate good $g_4$, such that $v_1(g_4) = 1$ and $v_2(g_4) = 5$. It is straightforward to see that if $\mathcal{A}$ assigns $g_4$ to agent 2, then the resulting allocation $(\{g_1\}, \{g_2, g_3, g_4\})$ is not $(1/2 + \varepsilon)$-\efo from the point of view of agent 1. Hence, we assume that $g_4$ is added to $A_1$ instead. The last good we need is $g_5$ with $v_1(g_5) = v_2(g_5) = 5$. Now, no matter who gets $g_5$, the resulting allocation, $(\{g_1, g_4, g_5\}, \{g_2, g_3\})$ or $(\{g_1, g_4\}, \{g_2, g_3, g_5\})$, is not $(1/2 + \varepsilon)$-\efo. To see this, first consider $(\{g_1, g_4, g_5\}, \{g_2, g_3\})$ and notice that from the point of view of agent 2, we have 
\[2 = v_2(A_2) < (1/2 + \varepsilon) \min_{g\in A_1} (v_2(A_1) - v_2(g)) = 3 +6 \varepsilon \,.\]
Finally, consider $(\{g_1, g_4\}, \{g_2, g_3, g_5\})$. From the point of view of agent 1, we have a similar situation:
\[2 = v_1(A_1) < (1/2 + \varepsilon) \min_{g\in A_2} (v_1(A_2) - v_1(g)) = 3 +6 \varepsilon \,.\]
In any case, algorithm $\mathcal{A}$ fails to maintain a $(1/2 + \varepsilon)$-\efo allocation within the first 5 time steps.
\end{proof}

By carefully inspecting the proof of Theorem \ref{thm:imposibility_0.5EF1}, one could notice that the same construction can be used to show that no algorithm can build an allocation that is $(1/3 + \varepsilon)$-\mms at every time step. In fact, the latter would imply Theorem \ref{thm:imposibility_0.5EF1} since it is known that any $(1/2 + \delta)$-\efo allocation for two agents is also a $(1/3 + \delta/9)$-\mms allocation (see, e.g., Proposition 3.6 of \citet{ABM18}). Nevertheless, for temporal maximin share fairness we can show a much stronger impossibility result that degrades with the number of agents. 

\begin{theorem}\label{thm:imposibility_MMS}
Let $\varepsilon > 0$ be a constant. There is no deterministic algorithm that, given a $2$-value instance with $n$ agents, always builds a $(1/(2n - 1) + \varepsilon)$-temporal-\mms allocation.  
\end{theorem}

\begin{proof}%[Proof of Theorem \ref{thm:imposibility_MMS}]
Suppose we have a deterministic algorithm $\mathcal{A}$ for the problem. We are going to consider a $2$-value instance with $n$ agents, such that, for all $i\in N$, $\beta_i = 1$ and $\alpha_i = \alpha = 2n^2 + 2n$.
Like in the proof of Theorem \ref{thm:imposibility_0.5EF1}, we begin with some straightforward, yet crucial observations. First, if at any point during the first $n$ time steps an agent receives a second good, by the time the first $n$ goods are fully allocated, there is at least one agent $j$ that has received value $0$ (because it got no goods) despite having a positive maximin share value $\bmu_j(n) \ge 1$. So, for what follows, we may assume that algorithm $\mathcal{A}$ assigns to each agent exactly one of the first $n$ goods. A second observation is that, given that $\ell < n$ goods have been already allocated, if the $(\ell+1)$-th good has the same value for all the agents who have not yet received a good, then it is without loss of generality to assume that $\mathcal{A}$ assigns it to the agent with the smallest index; this is just a matter of renaming the agents. 

Given these two observations above, suppose that goods $g_1, \ldots, g_n$ arrive, in this order, so that for any $i\in N$
\[v_i(g_r) = \begin{cases}
1\,, & \text{if } i\le r \\
\alpha\,, & \text{otherwise.}
\end{cases}\]
Then, algorithm $\mathcal{A}$ assigns them exactly as shown below, i.e., $g_i$ is given to agent $i$, for all $i\in N$. 
\begin{center}  
\begin{tabular}{lcccccc|cccc}
 & $g_1$ & $g_2$ & $g_3$ & \ldots & $g_{n-1}$ & $g_n$ & $g_{n+1}$ & $g_{n+2}$ & \ldots & $g_{2n-1}$  \\
ag.~1:  & \circled{1} & $\alpha$ & $\alpha$ &  \ldots & $\alpha$ & $\alpha$ & 1 & 1 & \ldots & 1  \\
ag.~2:  & 1 & \circled{1} & $\alpha$ &  \ldots & $\alpha$ & $\alpha$ & 1 & 1 & \ldots & 1  \\
ag.~3:  & 1 & 1 & \circled{1} &  \ldots & $\alpha$ & $\alpha$ & 1 & 1 & \ldots & 1  \\
$\vdots$  & $\vdots$ & $\vdots$ &  $\vdots$ & $\ddots$ & $\vdots$  & $\vdots$ & $\vdots$  & $\vdots$ & $\ddots$  & $\vdots$ \\
ag.~$n-1$:  & 1 & 1 & 1  & \ldots & \circled{1} & $\alpha$ & 1 & 1 & \ldots & 1  \\
ag.~$n$:  & 1 & 1 & 1 &  \ldots & 1 & \circled{1} & 1 & 1 & \ldots & 1  \\
\end{tabular}
\end{center}
At this point it is not hard to see that the next claim about goods that are low-valued for all holds. Notice that $1/2n < 1/(2n - 1) + \varepsilon$.

\begin{claim}\label{claim:high-valued}
Assume that the allocation of the first $n$ goods is as shown above. If at any point $t>n$ no agent $i\le \lambda$ has yet received total value more than $n+1$ and, furthermore, $\lambda$ goods, $g'_1, \dots, g'_{\lambda}$, which are high-valued for all agents arrive in that order, then either agent $j$ will get good $g'_j$, for all $j\in [\lambda]$, or algorithm $\mathcal{A}$ builds at some point an allocation that is not $1/2n$-\mms. 
\end{claim}

\begin{proof}[Proof of Claim \ref{claim:high-valued}]
\renewcommand\qedsymbol{{\scriptsize \textbf{Cl.~\ref{claim:high-valued}} $\boxdot$}}
This is a simple proof by induction on $\lambda$. For $\lambda = 1$, after $g'_1$ appears, say at time $t'$, we have $n_h(1, t') = n$, i.e., agent 1 has already seen $n$ high-valued goods (from its own perspective). So, we have $\bmu_1(t')\ge \alpha$. On the other hand, if $\mathcal{A}$ does not add $g'_1$ to $A_1$, we have $v_1(A_1)\le n+1 \le (n+1)\bmu_1(t')/\alpha = \bmu_1(t')/2n$.

Next, for $\lambda > 1$, assume that the claim is true for $\lambda - 1$ agents and high-valued goods. Now suppose that no agent $i\le \lambda$ has received value more than $n+1$ yet and that $\lambda$ goods, $g'_1, \dots, g'_{\lambda}$, which are high-valued for all agents arrive in that order. By this induction hypothesis, either algorithm $\mathcal{A}$ builds at some point an allocation that is not $1/2n$-\mms or, for all $j\in [\lambda-1]$, agent $j$ will get good $g'_j$. Given that, agent $\lambda$ has now seen $n$ high-valued goods in total and we can repeat the argument we made for agent 1 in the base case. Let $t''$ be the time step when $g'_{\lambda}$ arrives. We have $n_h(\lambda, t'') = n$ and, thus, $\bmu_{\lambda}(t'')\ge \alpha$, whereas $v_{\lambda}(A_{\lambda})\le n+1 \le \bmu_{\lambda}(t'')/2n$, unless $\mathcal{A}$ adds $g'_{\lambda}$ to $A_{\lambda}$. This completes the induction step.
\end{proof}

The next $n-1$ goods, $g_{n+1}, \ldots, g_{2n-1}$, are low-valued for all agents. Clearly, no matter how $\mathcal{A}$ assigns these goods, there is at least one agent, say $k\in N$, who does not get any of those by the end of time step $2n-1$. 
We will distinguish two cases, depending on whether agent $k$ is the last agent or not.
\medskip

\noindent\myul{\textbf{Case 1:} $k = n$.} In this case, the next $n-1$ goods, $g_{2n}, \ldots, g_{3n-2}$, are high-valued for all agents. Given that no agent has received total value more than $n$ by the end of time step $2n-1$, Claim \ref{claim:high-valued} applies, forcing the algorithm $\mathcal{A}$ to either fail or allocate $g_{2n}, g_{2n+1}, \ldots, g_{3n-2}$ to agents 1, 2, \ldots, $n-1$, in that order. We assume the latter. After this happens we notice that $\bmu_n(3n-2)= 2n-1$, since agent $n$ has already seen $n-1$ high-valued and $2n-1$ low-valued goods, and $\alpha> (2n-1)\cdot 1$. On the other hand, $A_n = \{g_n\}$, i.e., we have $v_n(A_n) = 1 = \bmu_n(3n-2)/(2n-1) <  (1/(2n - 1) + \varepsilon)\bmu_n(3n-2)$.
\medskip 

\noindent\myul{\textbf{Case 2:} $k < n$.} In this case, the next $n-k$ goods, $g_{2n}, \ldots, g_{3n-k-1}$, are such that, for $\ell \in [n-k]$ and for $i\in N$,
\[v_i(2n-1+\ell) = \begin{cases}
\alpha\,, & \text{if } i =  n-\ell+1\\
1\,, & \text{otherwise}
\end{cases}\]
i.e., $g_{2n}$ is high-valued only for agent $n$, $g_{2n+1}$ only for agent $n-1$, and so on, as is shown in the left part of the table below. Note, however, that we may not see the whole subsequence $g_{2n}, \ldots, g_{3n-k-1}$, depending on the behavior of algorithm $\mathcal{A}$.
\begin{center}  
\begin{tabular}{lcccc|cccc}
 &   $g_{2n}$ & $g_{2n+1}$  & \ldots & $g_{3n- k - 1}$ & $g_{3n- k}$ & $g_{3n- k+1}$  & \ldots & $g_{3n- 2}$\\
ag.~1:   & 1 & 1 & \ldots & 1 & \circled{$\alpha$} & $\alpha$& \ldots & $\alpha$ \\
ag.~2:   & 1 & 1 & \ldots & 1 & $\alpha$ & \circled{$\alpha$} & \ldots & $\alpha$ \\
$\vdots$   &$\vdots$  & $\vdots$ & $\ddots$  & $\vdots$ &$\vdots$  & $\vdots$ & $\ddots$  & $\vdots$ \\
ag.~$k-1$:    &1 & 1 & \ldots & 1 & $\alpha$ & $\alpha$ & \ldots & \circled{$\alpha$}  \\
ag.~$k$:    &1 & 1 & \ldots & 1 & $\alpha$ & $\alpha$ & \ldots & $\alpha$  \\
ag.~$k+1$:   & 1 & 1 & \ldots & \circled{$\alpha$} & $\alpha$ & $\alpha$ & \ldots & $\alpha$ \\
$\vdots$   &$\vdots$  & $\vdots$ & $\ddots$  & $\vdots$ &$\vdots$  & $\vdots$ & $\ddots$  & $\vdots$ \\
ag.~$n-1$:    & 1 & \circled{$\alpha$} & \ldots & 1 & $\alpha$ & $\alpha$ & \ldots & $\alpha$ \\
ag.~$n$:    &  \circled{$\alpha$} & 1 & \ldots & 1 & $\alpha$ & $\alpha$ & \ldots & $\alpha$ \\
\end{tabular}
\end{center}
We claim that either algorithm $\mathcal{A}$ fails to maintain a $(1/(2n - 1) + \varepsilon)$-\mms allocation or it allocates $g_{2n}, g_{2n+1}, \ldots, g_{3n-k-1}$ to agents $n$, $n-1$, \ldots, $k+1$, in that order. 
Towards a contradiction, suppose that this is not the case, i.e., $\mathcal{A}$ maintains a good approximation to temporal maximin share fairness, yet some of these goods are not allocated to the agent who values them highly. 
In particular, let $2n-1+j\in \{2n, 2n+1, \ldots, 3n-k-1\}$ be the lowest index of a good for which this happens. That is, $g_{2n}$ is given to agent $n$, $g_{2n+1}$ to agent $n-1$, and so on, $g_{2n-2+j}$ is given to agent $n-j+2$, \textit{but} $g_{2n-1+j}$ is not given to agent $n-j+1$. As a result, no agent $i\in[n-j+1]$ has received total value more than 
$n+1$ by the end of time step $2n-1+j$. So, at this point we may change the stream of goods, forget about $g_{2n+j},  \ldots, g_{3n-k-1}$ above, and replace them by $n-j$ goods, $\hat{g}_{2n+j}, \ldots, \hat{g}_{3n-1}$, which are high-valued for all agents. 
As we argued already, Claim \ref{claim:high-valued} applies here, forcing $\mathcal{A}$ to either fail or allocate $\hat{g}_{2n+j}, \hat{g}_{2n+j+1}, \ldots, \hat{g}_{3n-1}$ to agents 1, 2, \ldots, $n-j$, respectively, in that order. We assume the latter. After this happens we notice that $\bmu_{2n-1+j}(3n-1)\ge \alpha$, since agent $2n-1+j$ has already seen $n$ high-valued goods ($j-1$ among the first $n$ goods, the good $g_{2n-1+j}$ itself, and all of the last $n-j$ goods). On the other hand, $A_{2n-1+j}$ only contains at most $n+1$ low-valued goods. That is, we have $v_{2n-1+j}(A_{2n-1+j}) \le n+1 \le (n+1)\bmu_{2n-1+j}(3n-1)/\alpha <  \bmu_n(3n-1)/2n$.

At this point we may assume that algorithm $\mathcal{A}$ allocates $g_{2n}, g_{2n+1}, \ldots, g_{3n-k-1}$ to agents $n$, $n-1$, \ldots, $k+1$, in that order, as shown in the corresponding table. However, this means that none of the first $k$ agents has received more than $n$ low-valued goods so far. In particular, $A_k = \{g_k\}$. As a result we can consider $k-1$ additional goods, $g_{3n-k}, \ldots, g_{3n-2}$, which are high-valued for all agents and apply Claim \ref{claim:high-valued}. The claim guarantees that either algorithm $\mathcal{A}$ does not maintain a sufficiently high approximation guarantee throughout, or that the allocation will be completed as shown in the last part of the second table, i.e., $g_{3n-k}, g_{3n-k+1}, \ldots, g_{3n-2}$ are given to agents 1, 2, \ldots, $k-1$, respectively, in that order. This is a very poor allocation from the point of view of agent $k$. We notice that $\bmu_{k}(3n-2) = 2n - 1$, since agent $k$ has already seen $n-1$ high-valued goods ($n-k$ among the first $n$ goods and all of the last $k-1$ goods) and $2n-1$ low-valued goods ($k$ among the first $n$ goods and all the $(n-1)+(n-k)$ goods right after the first $n$ goods). On the other hand, recall that $A_k = \{g_k\}$. 
Thus, we have $v_{k}(A_{k}) = 1 = \bmu_n(3n-2)/(2n-1) <  (1/(2n - 1) + \varepsilon)\bmu_n(3n-2)$.
% \qed
\end{proof}

By inspecting the proof, it is not hard to see that Theorem \ref{thm:imposibility_MMS} could have been stated with respect to the largest ratio of the high over the low value of an agent. This is particularly relevant for Section \ref{sec:beyond_2value}, where this term will appear in the fairness guarantees we obtain for personalized interval-restricted instances.

\begin{corollary}\label{cor:imposibility_MMS}
Let $\varepsilon > 0$ be a constant. There is no algorithm that, given a $2$-value instance with $n$ agents and values $\alpha>1 =\beta$, always builds a $(1/\sqrt{2 \alpha} + \varepsilon)$-temporal-\mms allocation.  
\end{corollary}

\begin{proof}
Notice that in the proof of Theorem \ref{thm:imposibility_MMS} we have  $ \alpha = 2n^2 + 2n$. Also, by standard calculus we get $\lim_{n\to\infty} \Big(\frac{1}{2n-1} - \frac{1}{\sqrt{4n^2 + 4n}}\Big) = 0$. That is, for large enough $n$, it holds that $\frac{1}{\sqrt{4n^2 + 4n}} + \varepsilon > \frac{1}{{2n-1}} + \frac{\varepsilon}{2}$ and the impossibility follows directly by Theorem \ref{thm:imposibility_MMS}.
\end{proof}

%%%%%%%%%%%%%%%%%%%%%%%%%%%%%%%%%%%%%%%%%
%%%%%%%%%%%%%%%%%%%%%%%%%%%%%%%%%%%%%%%%%
\section{A Tight Algorithm}
\label{sec:main_algorithm}
%%%%%%%%%%%%%%%%%%%%%%%%%%%%%%%%%%%%%%%%%
%%%%%%%%%%%%%%%%%%%%%%%%%%%%%%%%%%%%%%%%%

In this section we present the main result of this work, an algorithm with tight temporal maximin share fairness guarantees. Given how nuanced the construction of the example in the proof of  Theorem \ref{thm:imposibility_MMS} was, it is not particularly surprising that  matching this $1/(2n-1)$ factor requires a fairly elaborate algorithm that performs careful book-keeping of who should get the next contested high-valued goods. 

Before discussing any details of our Deferred-Priority algorithm (Algorithm \ref{alg:main_alg}), we revisit the observation that we  made at the beginning of the proof of Theorem \ref{thm:imposibility_MMS}. 
When agents have positive values, if we aim for a nonzero temporal maximin share fairness guarantee, Algorithm \ref{alg:main_alg} should assign to each agent exactly one of the first $n$ goods. Indeed, if this is not the case,  there would be at least one agent $j$ who receives value $0$ (because $j$ got no goods) at the end of time step $n$, despite having a positive maximin share value $\bmu_j(n) \ge \beta_i = 1$. So, no matter how our algorithm works in general, the allocation of the first $n$ goods is ``special'' in the sense that it allows for extremely little flexibility. We call these first $n$ time steps Phase $0$. More generally, Algorithm \ref{alg:main_alg} will operate in phases; we want to ensure that during each phase every agent gets at least one good and the phases do not last for too long.

As a second general design goal, however, we want to allocate goods to agents who consider them high-valued as frequently and as uniformly as possible. Note that this is rather incompatible with the aforementioned goal of frequently giving goods to everyone. Our solution to that is to have two sets of counters (the entries of the vectors $H$ and $L$, introduced in line \ref{line:initialization}) that keep track of how many high- or low-valued goods, respectively, an agent affords to lose to others before we are in a situation where a ``bad'' sequence of goods inevitably destroys the temporal maximin share fairness guarantee. 

So, the general idea is that throughout the allocation the agents should gain higher priority the more value they lose to others; the corresponding  \textit{priority levels} are implicitly described by the entries of $H$ and $L$ (and later on explicitly defined for high-valued goods as $\mathcal{H}_{\ell}(t)$). Indeed, every time a good arrives and is allocated, the entries of $H$ and $L$ are updated accordingly.  The way these indices are updated enforces that one out of every $2n-1$ goods, in general, and one out of $n$---that later becomes $3n-2$---high-valued goods is allocated to each agent. These quantities may seem somewhat loose, but as it shown in our analysis, asking for more frequent allocations per agent, would not leave enough room for our competing goals to work simultaneously.

In our proofs and statements we often need to refer to the agents' bundles at different time steps. For clarity, we write $A_i^t$ (rather than just $A_i$) to denote the bundle of agent $i$ at the end of time step $t$. In fact, we use this notation to the statement of the main theorem of this section as well.  Recall also that $n_h(i, t) = |\{ g\in M \,|\, v_i(g) = \alpha_i \text{ and } g \text{ is one of the first $t$ goods that arrived}\}|$ is the number of high-valued goods agent $i$ has seen up to (and incuding) time $t$.

% \makeatletter
% \renewcommand{\ALG@name}{Protocol}
% \makeatother

\begin{algorithm}[h!t]
		\caption{Deferred-Priority$(v_1, \ldots, v_n; M)$ \\{\small {(The valuation functions, $v_i$, $i\in [n]$, are given via oracles; $M$ is given in an online fashion, one good at a time.)}}}
		\begin{algorithmic}[1]
            \State $\mathtt{phase} \leftarrow 0$\textbf{;} $\mathtt{low} \leftarrow 0$\textbf{;} $\mathtt{high} \leftarrow 0$  \textbf{;} $t\leftarrow 0$ \Comment{We initialize all of our counters.} %\textbf{;} $\mathtt{flag\_high} \leftarrow 0$
			\For{$i \in N$}
                \State $A_i \leftarrow \emptyset$\textbf{;} $H[i] \leftarrow n$\textbf{;} $L[i] \leftarrow 2n-1$\textbf{;} $\chi_i \leftarrow 0$ \Comment{Ag.~$i$ tolerates the loss of less than $n$ high-valued goods.} \label{line:initialization}
            \EndFor\vspace{1ex}
            \State \textbf{whenever a new good $g$ arrives:} 
            \State $t\leftarrow t + 1$
            \For{$i \in N$}
                \If{$v_i(g) = \alpha_i > 0$}
                    \State $H[i] \leftarrow H[i] - 1$ \label{line:reduce_H} \Comment{Potential loss of a high-valued good; $i$'s priority for high-valued goods is increased.}
                    % \State $\mathtt{flag\_high} \leftarrow 1$
                \Else
                    \State $L[i] \leftarrow L[i] - 1$ \Comment{Potential loss of a low-valued good; $i$'s priority for low-valued goods is increased.}
                    % \State $\mathtt{low} \leftarrow \mathtt{low} + 1$
                \EndIf
            \EndFor
            \State $N_h(g, t) \leftarrow \{i\in N \,|\, v_i(g) = \alpha_i  \text{ and } \chi_i = 0\}$ \label{line:N_h} \Comment{Potential recipients of $g$ as a high-valued good.} 
            \State $N_{\ell}(g, t) \leftarrow \{i\in N \,|\, v_i(g) = \beta_i \text{ and } \chi_i = 0\}$ \label{line:N_ell} \Comment{Potential recipients of  $g$ as a low-valued good.}
            \If{$N_h(g, t) \neq \emptyset$}\label{line:check_Nh} 
                \State $\mathtt{high} \leftarrow \mathtt{high} + 1$ \label{line:high_increased} \Comment{$g$ is allocated as a high-valued good.}
                \State $j = \argmin_{i \in N_h(g, t)} H[i]$ \label{line:choice_of_j} \Comment{Ag.~$j$ has the highest priority for $g$; ties are broken lexicographically.}
                \State $A_j \leftarrow A_j \cup \{g\}$ \Comment{The good is added to ag.~$j$'s bundle.}
                \State $H[j] \leftarrow H[j] + 3n-2$ %\ga{Actually, $3n-2$ suffices for the proofs and makes things a bit nicer for the corollaries.}  
                \label{line:increase_H} \Comment{Now ag.~$j$ tolerates the loss of at most $3n-2$ high-valued goods.}
                \State $\chi_j \leftarrow 1$ \label{line:chi_increased}\Comment{Ag.~$j$ will not get any more goods during the current phase.} 
            \Else
                \State $\mathtt{low} \leftarrow \mathtt{low} + 1$ \label{line:low_increased} \Comment{$g$ is allocated as a low-valued good.}
                \State $j = \argmin_{i \in N_{\ell}(g, t)} L[i]$ \Comment{Ag.~$j$ has the highest priority for $g$; ties are broken lexicographically.}
                \State $A_j \leftarrow A_j \cup \{g\}$  \Comment{The good is added to ag.~$j$'s bundle.}
                \State $L[j] \leftarrow  2n + t$ \label{line:last_in_line} \Comment{Ag.~$j$ now has the lowest priority among active agents for the rest of the phase.} 
                \If{$\mathtt{phase} = 0$} \label{line:chi_increased_0_condition}
                    \State $\chi_i \leftarrow 1$ \Comment{If this is phase $0$, ag.~$j$ will not get any more goods.} \label{line:chi_increased_0}
                \EndIf
            \EndIf
            \If{($\mathtt{phase} = 0$ \textbf{and} $\mathtt{low}+\mathtt{high} = n$) \textbf{or} ($\mathtt{phase} > 0$ \textbf{and} $\max\{\mathtt{low},\mathtt{high}\} = n$) \label{line:phase_condition}} %\\\hspace{3pt}\Comment{The conditions to conclude this phase are met.} 
                \State $\mathtt{phase} \leftarrow \mathtt{phase} + 1$ \Comment{The conditions to conclude this phase were met and we move to the next one.} 
                \State $\mathtt{low} \leftarrow 0$\textbf{;} $\mathtt{high} \leftarrow 0$ \Comment{We reset our counters.}
                % \textbf{;} $\mathtt{flag\_high} \leftarrow 0$  
                \For{$i \in N$}
                    \State $L[i] \leftarrow 2n-1$ \Comment{We reset the priority for low-valued goods.} 
                \EndFor
            \EndIf
		\end{algorithmic}
		\label{alg:main_alg}
\end{algorithm}

\pagebreak

\begin{theorem}\label{thm:main_algorithm}
Algorithm \ref{alg:main_alg} builds an allocation such that, for every $i\in N$: \begin{enumerate}[leftmargin=20pt,itemsep=3pt]
    \item $|A_i^n| = 1$, i.e., agent $i$ gets one of the first $n$ goods (assuming $m\ge n$; otherwise $|A_i^m| \le 1$). 
    \item At (the end of) any time step $t$, $A_i^t$ contains at least $\lfloor n_h(i, t) / (3n-2)  \rfloor$ high-valued goods. Moreover, if agent $i$ gets to see at least $n$ high-valued goods, at (the end of)  time step $t_{i0} = \min\{t \,|\, n_h(i, t)\ge n\}$, $A_i^{t_{i0}}$ contains at least $1$ high-valued good.
    \item At (the end of) any time step $t\ge n$, $|A_i^t| \ge\lfloor (t-n) / (2n - 1) \rfloor +1$, i.e., agent $i$ has  received at least one out of every $2n - 1$ goods they have seen after the first $n$ goods.
\end{enumerate}    
\end{theorem}

\begin{proof}[Proof of parts 1.~and 3.]
We are going to show parts 1., 2., and 3.~separately. While parts 1.~and 3.~are relatively straightforward, part 2.~requires an elaborate analysis using delicate inductive arguments. During any phase, an agent $i$ is called \emph{active} if $\chi_i = 0$ and \emph{inactive} otherwise. As it is clear by the definition of the sets $N_h(g, t)$ and $N_{\ell}(g, t)$ (lines \ref{line:N_h} and \ref{line:N_ell}), no agent can receive any more goods during the current phase once it becomes inactive. 

We begin with part 1. of the theorem.
Lines \ref{line:chi_increased} and \ref{line:chi_increased_0_condition}-\ref{line:chi_increased_0} update $\chi_i$ to $1$ for any agent who receives \textit{any} good during Phase $0$, ensuring that no one gets more than $1$ good during this phase. Further, if $m\ge n$, lines \ref{line:high_increased} and \ref{line:low_increased}, combined with the \textit{first} part of the condition in line \ref{line:phase_condition}, ensure that exactly $n$ goods are allocated during Phase $0$, as either low-valued or high-valued goods. Therefore, $|A_i^n| = 1$, i.e., each agent gets exactly $1$ of the first $n$ goods.

Moving to part 3., let $q\in \mathbb{Z}^+$. We observe that, during Phase $q$, no more than $2n-1$ goods are allocated, as it is enforced by lines \ref{line:high_increased} and \ref{line:low_increased}, combined with the \textit{second} part of the condition in line \ref{line:phase_condition}. Next, note that, like in Phase $0$, if an agent $i$ receives a high-valued good (which triggers $\chi_i$ to become $1$ in line \ref{line:chi_increased}), becomes inactive and never receives another good during Phase $q$. However, unlike in Phase $0$, when agent $i$ receives a low-valued good at time $t$, it now stays active. Nevertheless, in such a case, agent $i$ becomes \textit{last} in the implicit priority list of active agents for low-valued goods, as now $L[i] = 2n + t$ (line \ref{line:last_in_line}) and $L[j] \le 2n + t'$ with $t'<t$, for $j\in N\setminus \{i\}$. The important observation here is that once this happens, agent $i$ cannot receive another good before every other agent is inactive or receives a low-valued good at a time $t''>t$. We are now ready for the main argument that implies part 3.

If Phase $q$ terminates because $\mathtt{high} = n$, then every agent received exactly one high-valued good. Otherwise, if Phase $q$ terminates because $\mathtt{low} = n$, we distinguish two simple cases. If no agent got more than one low-valued good in Phase $q$, then everyone received exactly one low-valued good (and at least one good, in general). So, assume that there is some agent $i$ who received at least $2$ low-valued goods during Phase $q$. By the preceding discussion, this means that every other agent became inactive at some point (by receiving a high-valued good) or received a low-valued good between the times when agent $i$ got its first and second low-valued good. In any case, if Phase $q$ terminates, then each agent has received at least $1$  out of at most $2n-1$ goods that were allocated during the phase. 
Combining this with the fact that every agent gets $1$ out of $n$ goods in Phase $0$ (by part 1.), we  
conclude that at the end of any time step $t$, $|A_i^t|\ge 1+ \lfloor (t -n)/ (2n - 1) \rfloor$.
\end{proof}

Most of the remaining section is dedicated to proving part 2.~of Theorem \ref{thm:main_algorithm}. Thus, it would be useful to give some intuition behind both the statement and its proof. In doing so, we will establish some additional notation and terminology. The obvious way to show that each agent gets at least $1$ out of the first $n$ high-valued goods they see and $1$ out of every $3n-2$ high-valued goods overall, is to show that it is always possible to allocate the goods in such a way so that $H[i]>0$, for all $i\in N$, at the end of each time step $t$ (i.e., right before the $(t+1)$-th good arrives). The reason, of course, is that $H[i]$ has been defined to contain the number of high-valued goods that agent $i$ can afford to lose before the desideratum of part 2.~of Theorem \ref{thm:main_algorithm} is violated. 

A straightforward necessary condition for $H[i]>0, i\in N$, to hold at the end of each time step $t$ is to have at most one agent $j$, such that $H[j] = 1$ before a new good arrives. To see this, assume that there are distinct $j, j'$ such that $H[j] = H[j'] = 1$ and that the next good $g$ that arrives is high-valued for everyone. Then, no matter how $g$ is allocated, at least one of $H[j], H[j']$ will hit $0$, meaning that enough high-valued goods were lost for one of the two agents for part 2.~of Theorem \ref{thm:main_algorithm} to fail. In fact, we can extend this necessary condition to having at most $k$ agents $j_1, \ldots, j_k$, such that $H[j_{\ell}] \le k$ before a new good arrives, for any $k\in[n]$. Indeed, if there are at least $k+1$ such agents and the next $k$ goods are high-valued for everyone, it is impossible to allocate them without making some coordinate(s) of vector $H$ equal to $0$.

The interesting thing here, is that this simple necessary condition always allows us to ``legally'' allocate at least one next good $g$. Roughly speaking, if we consider the agents who see $g$ as high-valued and give it to the agent with the smallest $H$ entry among them, 
then it is not very hard to see (we will prove it formally in Claim \ref{claim:sufficient_condition}) that it is not possible to end up with any agent $i$ having $H[i] = 0$ (recall that after receiving a good, an agent's $H$ entry increases significantly). So, if one could allocate each good and, at the same time, maintain the condition that $H$ contains at most $k$ entries that are $k$ or below, for all $k\in [n]$, then part 2.~of Theorem \ref{thm:main_algorithm} would follow. 
The tricky part is to make sure that this condition still holds after allocating \textit{any} good and this is the core of the technical difficulty of proving the theorem, mainly because we often need to allocate goods that are not low-valued for everyone to agents who see them as low-valued. In fact, the latter is absolutely necessary for part 3.~of the theorem shown above.

In order to formalize things, we introduce the following notation for the \textit{level sets} that contain all agents with the same priority according to $H$ at any given time:
\[\mathcal{H}_{\ell}(t) = \{i \in N \,|\, H[i] = \ell \text{ at the end of time step } t\}\,.\]
With this notation, the above necessary and sufficient condition for being able to legally extend the allocation at time $t+1$ becomes
\begin{equation}\label{eq:main_condition}
	\bigg| \bigcup_{\ell=0}^{k} \mathcal{H}_{\ell}(t) \bigg| \le k, \text{ for all } 0\le k \le n \,.
\end{equation}

In the discussion above, we imply that maintaining \eqref{eq:main_condition} is easier if, whenever a good is viewed as high-valued by someone, it is always allocated as a  high-valued good. Indeed, this is the case: if we only allocated goods so as to maximize the social welfare, then  we would be able to maintain \eqref{eq:main_condition} for every $t$, even if in line \ref{line:increase_H} we only added $n$ rather than $3n-2$. The technical reason why will become clear in the proofs of Claims \ref{claim:H-condition_<=n} and \ref{claim:H-condition_>n}, but the issue with this is that it would mean that agents who mostly see low-valued goods might have to wait arbitrarily long before getting anything. Hence, we add $3n-2$ in line \ref{line:increase_H}, to give us some extra room to keep every agent content. From a technical point of view, within our proofs we typically need to decouple the two cases that cause changes to entries of $H$ (allocating a good that is not globally low-valued as high-valued \textit{versus} as low-valued), as they are qualitatively very different.

The following lemma states the fact that condition \eqref{eq:main_condition} holds for all time steps $t\ge 0$. As its proof is fairly long and complicated, it is deferred to the Section \ref{subsec:G-lemma_proof}. 

\begin{lemma}\label{lemma:H-condition}
	At the end of any time step $t\ge 0$, we have $\big| \bigcup_{\ell=0}^{k} \mathcal{H}_{\ell}(t) \big| \le k$, for all $0\le k\le n$.
\end{lemma}

At this point we are ready to prove part 2.~of Theorem \ref{thm:main_algorithm}.

\begin{proof}[Proof of part 2.~of Theorem \ref{thm:main_algorithm}]
In the discussion preceding this proof, we claimed that condition \eqref{eq:main_condition}
is sufficient, and we briefly argued about it in a rather hand-wavy way. Here we begin by formalizing this fact. It should be noted that Claim \ref{claim:sufficient_condition} does not say that by allocating the $(t+1)$-th good $g$ we ensure that the condition holds for time step $t+1$; this is shown separately in Lemma \ref{lemma:H-condition}. The claim barely states that whenever condition \eqref{eq:main_condition} holds and things have gone well in the past, Algorithm \ref{alg:main_alg} can allocate the next good without violating the guarantees we aim to show. 

\begin{claim}\label{claim:sufficient_condition}
Let $t$ be any time step in $\{0, 1, \ldots, m-1\}$ and let $g$ be the $(t+1)$-th good. Assuming that all entries of $H$ have remained positive at the end of all time steps up to $t$, condition \eqref{eq:main_condition} guarantees that Algorithm \ref{alg:main_alg} can allocate $g$ without any entry of $H$ becoming $0$ at the end of time step $t+1$.
\end{claim}

\begin{proof}[Proof of Claim \ref{claim:sufficient_condition}]
\renewcommand\qedsymbol{{\scriptsize \textbf{Cl.~\ref{claim:sufficient_condition}} $\boxdot$}}
For the sake of clarity, here we make the dependency of $H$ on $t$ explicit and write $H_t[i]$ to denote $H[i]$ at the end of time step $t$.
Assume that $\big| \bigcup_{\ell=0}^{k} \mathcal{H}_{\ell}(t) \big| \le k$, for all $k\in [n]$. In particular, $| \mathcal{H}_1(t) | \le 1$. If $| \mathcal{H}_1(t) | = 0$, then $H_t[i]\ge 2$ for all $i\in N$ and clearly, no matter how good $g$ is allocated, $H_{t+1}[i]\ge 1$ for all $i\in N$. Similarly, if $| \mathcal{H}_1(t) |  = 1$ and $j\in N$ is the unique agent such that $H_{t}[j] =  1$ but $v_j(g) = \beta_j$, we have $H_{t+1}[j] = H_{t}[j] = 1$ as well as $H_{t+1}[i]\ge H_{t}[i] - 1\ge 1$ for all $i\in N\setminus\{j\}$, like before, no matter how $g$ is allocated. 

The interesting case here is when $|\mathcal{H}_1(t)|  = 1$,  $j\in N$ is the unique agent such that $H_{t}[j] =  1$ and $v_j(g) = \alpha_j$. Again, for $i\in N\setminus\{j\}$, $H_{t+1}[i]\ge 1$ but now we must show that Algorithm \ref{alg:main_alg} will add $g$ to $A_j$. Given that $j$ has the highest priority (i.e., lowest entry in $H$, which has temporarily dropped to $0$), for $j$ to get $g$ it suffices to show that $\chi_j = 0$. Towards a contradiction, assume this is not the case, i.e., $j$ has received a high-valued good at time $t'$ which belongs to the same phase as $t$. Since each phase has at most $2n-1$ time steps (see the proof of part 3.~of Theorem \ref{thm:main_algorithm}), $t-t' \le 2n-2$. But then, 
\[H_{t}[j] \ge H_{t'}[j] - (2n-2) = H_{t'-1}[j] - 1 + 3n -2 -2n +2 \ge n \,,\]
contradicting the choice of $j$ . We conclude that an agent $j$ with $H_{t}[j] =  1$ cannot have received a high-valued good in the phase that includes $t$, thus $\chi_j = 0$. Therefore, $j$ is the agent who gets $g$ in line \ref{line:choice_of_j}, and $H_{t+1}[j] = H_{t}[j] - 1 + 3n -2= 3n -2 > 0$, completing the proof.
\end{proof}

Recall that, by the design of the priority vector $H$, in order to show that each agent gets at least $1$ out of the first $n$ high-valued goods they see and $1$ out of every $3n-2$ high-valued goods overall, it suffices  to show that it is always possible to allocate the goods so that $H[i]>0$, for all $i\in N$, at the end of each time step $t$. 
By combining Claim \ref{claim:sufficient_condition} with Lemma \ref{lemma:H-condition}, which shows that condition \eqref{eq:main_condition} is maintained throughout the execution of Algorithm \ref{alg:main_alg}, we have exactly  that. The algorithm allocates all goods without any entry of $H$ becoming $0$ at the end of any time step. Equivalently, by the definition of how $H$ is updated, agent $i$ receives at least $1$ high-valued good by time $t_{i0} = \min\{t \,|\, n_h(i, t)\ge n\}$ and at least $1$ out of every $3n-2$ high-valued goods they see after that, for a total of at least $\lfloor n_h(i, \tau) / (3n-2)  \rfloor$ high-valued goods by the end of a time step $\tau\ge 0$.
\end{proof}

Now, using Theorem \ref{thm:main_algorithm}, we can argue about the temporal maximin share guarantees of the Deferred-Priority algorithm (Algorithm \ref{alg:main_alg}). Recall from Section \ref{sec:prelims}
that an agent $i$ is of type 1 when $\alpha_i > \beta_i =1$, it is of type 2 when $\alpha_i = \beta_i =1$ and it is of type 3 when $1 = \alpha_i > \beta_i = 0$. As the arguments needed from different types are somewhat different, we will state the corresponding guarantees separately. 

\begin{corollary}\label{cor:types2-3-4}
Any agent $i$ of type $k\in \{2, 3\}$ receives at least a constant fraction of its temporal maximin share by Algorithm \ref{alg:main_alg}. In particular, $v_i(A_i^t) \ge \bmu_i(t) / k$, for any time step $t\ge 0$. 
\end{corollary}

\begin{proof}
First let agent $i$ be of type 2. We will bound the value $i$ gets during the allocation sequence induced by Algorithm \ref{alg:main_alg}. At the end of any time step $t<n$, we have $\bmu_i(t) = 0$, so the statement trivially holds. 
If, instead, $n\le t < n+ (2n-1)$, then by part 1.~of Theorem \ref{thm:main_algorithm} we have $|A_i^t|\ge 1$ and so, $v_i(A_i^t)\ge 1$, whereas $\bmu_i(t) \le \lfloor (3n-1)/n \rfloor = 2$. 
Finally, we may assume that $n+k(2n-1) \le t < n+ (k+1)(2n-1)$, for some $k\in \mathbb{Z}_{>0}$. Then, by part 3.~of Theorem \ref{thm:main_algorithm}, agent $i$ has received at least one out of every $2n-1$ goods they have seen after Phase $0$ and  so, $v_i(A_i^t)=|A_i^t|\geq \lfloor (t-n)/(2n-1) \rfloor + 1 = k + 1$, whereas, by the definition of maximin share, 
\[\bmu_i(t) = \left\lfloor \frac{t}{n} \right\rfloor < \left\lfloor \frac{n+ (k+1)(2n-1)}{n} \right\rfloor = \left\lfloor \frac{2(k+1)n + (n-k-1)}{n} \right\rfloor \le 2(k+1)\,.\]

Next, assume that agent $i$ is of type $3$. Given that low-valued goods are completely irrelevant to agent $i$,  we can consider a time alternative $\tau$ that starts at $0$, like $t$, but only increases when a high-valued good for $i$ arrives. That is, while $t$ reflects how many goods have arrived in general, $\tau$ reflects how many high-valued goods with respect to agent $i$ have arrived instead. We can repeat the exact same analysis we did for agents of type 2, but using $\tau$ instead of $t$, $3$ instead of $2$ for the factor,  $3n-2$ whenever the quantity $2n-1$ was used, and by invoking part 2.~of Theorem \ref{thm:main_algorithm} instead of part 3. %\qed  
\end{proof}

\begin{corollary}\label{cor:type1}
Any agent $i$ of type 1 receives at least a $1/(2n-1)$ fraction of its temporal maximin share by Algorithm \ref{alg:main_alg}., i.e., $v_i(A_i^t) \ge \bmu_i(t) /(2n-1)$, for any time step $t\ge 0$, which improves to $\Omega(1)$ from time $t_{i0}$ onward (recall that $t_{i0} = \min\{t \,|\, n_h(i, t)\ge n\}$). 
\end{corollary}

\begin{proof}
Let $i$ be a type 1 agent and consider any time step $t$. Let  $\kappa, \lambda$ be the number of high-valued and low-valued goods in $A_i^t$, respectively, where $\kappa, \lambda\in \mathbb{Z}_{\ge0}$.
\smallskip

\noindent\myul{\textbf{Case 1:} $\kappa = 0$.} 
At the end of any time step $t<n$, we have $\bmu_i(t) = 0$ and the statement trivially holds. So assume that $t\ge n$
As $\kappa = 0$ implies that $n_h(i, t) \le n-1$, or equivalently, $t<t_{i0}$ (by part 2.~of Theorem \ref{thm:main_algorithm}), we have $n\le t <t_{i0}$. On one hand, we know that $v_i(A_i^t) = \lambda \ge 1$, where the second inequality follows by part 1.~of Theorem \ref{thm:main_algorithm}. On the other hand, by considering a hypothetical allocation where each one of the $n_h(i, t)$ high-valued goods for agent $i$ is a whole bundle and the low-valued goods for agent $i$ are split as equally as possible into $n-n_h(i, t)$ bundles, we see that $i$'s maximin share is at most the value of the worst bundle among the ones filled with low-valued goods, i.e., 
\[\bmu_i(t) \le \left\lfloor \frac{t-n_h(i, t)}{n-n_h(i, t)}\right\rfloor \le \left\lfloor \frac{t-n_h(i, t) - (n-1-n_h(i, t))}{n-n_h(i, t) - (n-1-n_h(i, t))}\right\rfloor \le \left\lfloor \frac{t - (n-1)}{1}\right\rfloor = t - (n-1)\,,\]
where the second inequality follows from $n_h(i, t) \le n-1$ and the simple fact that $\frac{a}{b} \le \frac{a-c}{b-c}$ for any $a \ge b > c \ge 0$. Note however that there is a straightforward upper bound on $t$, implied by parts 1.~and 3.~of Theorem \ref{thm:main_algorithm}: $t \le  n+ (\lambda - 1)(2n-1) + (2n - 2)$. Thus, 
\[\bmu_i(t) \le n+ (\lambda - 1)(2n-1) + (2n - 2) - (n-1) = \lambda (2n - 1) = (2n - 1) v_i(A_i^t)\,.\]
% \smallskip

For the remaining two cases, we are going to show an approximation factor of at least $1/4$ with respect to proportionality, which then implies the same guarantee for maximin share fairness. Let $S$ be the set containing the first $t$ goods for some $t$. Note that necessarily one of these two cases holds if $t\ge t_{i0}$ (but possibly even earlier than that).\smallskip

\noindent\myul{\textbf{Case 2:} $\kappa \geqslant 1$ and $\lambda = 0$.} 
In this easier case we have $v_i(A_i^t) = \kappa \alpha_i $ and $i$ might have seen at most $\kappa (2n -1) + (2n-2)$ goods (by part 3.~of Theorem \ref{thm:main_algorithm}) of total value $v_i(S) \le (\kappa (2n -1) + (2n-2))\alpha_i$. We have
\[\bmu_i(t)\le \frac{v_i(S)}{n} \le \frac{\kappa (2n -1) + (2n-2)}{n} \alpha_i \le \frac{2n\kappa_i  + 2n}{n} \alpha_i = 2(\kappa +1)\alpha_i \le 4\kappa \alpha_i = 4 v_i(A_i^t)\,.\]
%\smallskip

\noindent\myul{\textbf{Case 3:} $\kappa \geqslant 1$ and $\lambda \geqslant 1$.} Similarly to Case 2, $v_i(A_i^t) = \kappa \alpha_i + \lambda$
and $i$ might have seen at most $n + (\kappa + \lambda - 1) (2n -1) + (2n-2)$ goods (by parts 1.~and 3.~of Theorem \ref{thm:main_algorithm}), out of which at most $n + (\kappa  - 1) (3n -2) + (3n-3)$ can be high-valued for $i$ (by parts 1.~and 2.~of Theorem \ref{thm:main_algorithm}). Then it is a matter of simple calculations to show that for the total value $v_i(S)$ to be maximized, the low-valued goods are at most $(\lambda - 1) (2n -1) + (2n-2)$. That is,  
\begin{IEEEeqnarray*}{rCl}
	v_i(S) & \le & [n + (\kappa  - 1) (3n -2) + (3n-3)]\alpha_i + [(\lambda - 1) (2n -1) + (2n-2)]	\\
	& \le & [n + 3n(\kappa  - 1)  + 3n]\alpha_i + 2n(\lambda - 1) + 2n =  3n\Big(\kappa+ \frac{1}{3}\Big)\alpha_i +2n \lambda  \,.
\end{IEEEeqnarray*}
Therefore, we have 
$\bmu_i(t)\le {v_i(S)}/{n} \le (3\kappa +1) \alpha_i +2 \lambda \le 4\kappa  \alpha_i +4 \lambda =  4 v_i(A_i^t)$.
\end{proof}

%%%%%%%%%%%%%%%%%%%%%%%%
\subsection{Proving Lemma \ref{lemma:H-condition}}\label{subsec:G-lemma_proof}
%%%%%%%%%%%%%%%%%%%%%%%%

\begin{proof}[Proof of Lemma \ref{lemma:H-condition}]
As we did in the proof of Claim \ref{claim:sufficient_condition}, for clarity, we write $H_t[i]$ to denote $H[i]$ at the end of time step $t$.
The proof will be broken down into two proofs, one for $t\le n$ and one for $t\ge n$; for the sake of presentation, the corresponding cases of the lemma are stated as Claims \ref{claim:H-condition_<=n} and \ref{claim:H-condition_>n} below.

\begin{claim}\label{claim:H-condition_<=n}
	At the end of any time step $t\le n$, we have $\big| \bigcup_{\ell=0}^{k} \mathcal{H}_{\ell}(t) \big| \le k$, for all $0\le k \le n$.
\end{claim}

\begin{proof}[Proof of Claim \ref{claim:H-condition_<=n}]
\renewcommand\qedsymbol{{\scriptsize \textbf{Cl.~\ref{claim:H-condition_<=n}} $\boxdot$}}
We will use induction on the number of agents $n$ for a slightly more general version of the algorithm that takes the initialization of $H = H_0$ as part of the input, where it must be that $H_0[i] \ge n$ for all $i\in N$. Then the statement of the claim follows by fixing this part of the input to be $H_0[i] = n$ for all $i\in N$.  

For a single agent, it is straightforward that, for $H_0[1]\ge 1$, initially $\big| \mathcal{H}_0(0) \big| = 0$ and $\big| \mathcal{H}_1(0) \big| \le 1$, %
whereas after the first good is allocated, either $H_1[1]$ remains unchanged and, thus, at least $1$ (if the good was low-valued) or it is updated to $H_1[1] - 1 + 3\cdot 1 -2  \ge 1$ (if the good was high-valued). Either way, $\big| \mathcal{H}_0(1) \big| = 0$ and $\big| \mathcal{H}_1(1) \big| \le 1$, completing our base case.

Now assume that the statement of the claim is true for a certain $n'\ge 1$, and consider any instance with $n=n' +1$ agents and any $H_0\in \mathbb{Z}_{\ge n}^n$. For this particular instance, let $g_1, g_2, \ldots, g_m$ be the  goods that arrive, in this order. 
Because of how goods are allocated in Phase $0$, i.e., no agent gets a second good, once an agent $j$ receives $g_1$ at time $t = 1$,  the remaining of Phase $0$ is indistinguishable from (the complete) Phase $0$ of an instance with the agents of $N\setminus \{j\}$,  an appropriate initial priority vector (defined by $H_0[i]-\bm{1}_{N_h(g_1, 1)}(i)$ or $H_0[i]-\bm{1}_{N_h(g_1, 1)}(i) -1$; see Cases 1 and 2 below), and the sequence of goods being $g_2, g_3, \ldots, g_m$. 
We are going to invoke the induction hypothesis on this sub-instance but we distinguish two cases, depending on whether $g_1$ is allocated as a high-valued or a low-valued good.
Notice that it is without loss of generality to assume that the agent who gets good $g_1$ is agent $1$, as it is a matter of renaming the agents, if needed, and is consistent with our lexicographic tie-breaking. 

For the sub-instance that only involves agents $2$ through $n$, has a properly defined initial priority vector $H'_0$ (see within the cases for the corresponding description), and the sequence of goods is $g_2, \ldots, g_m$, we use a prime to distinguish the corresponding quantities. 
That is, we use $t'$ to denote time, rather than $t$ that we reserve for the original instance; in general $t' = t-1$, e.g., the $5$th time step in the sub-instance corresponds to the $6$th time step of the original problem. Thus, we will write $H'_{t'}$ for the  priority vector of the sub-instance at the end of time step $t'$ \emph{of that instance} and $\mathcal{H}'_{\ell}(t')$ for the level sets that it induces.
Assuming that $H'_0$ is such that $H'_0[i]\ge n-1$ for $i\in \{2, \ldots, n\}$, by the induction hypothesis, we have for this sub-instance:
\begin{equation}\label{eq:IH_G-condition_<=n}
\text{At the end of any } t'\le n-1, \, \bigg| \bigcup_{\ell=0}^{k} \mathcal{H}'_{\ell}(t') \bigg| \le k  \text{, for all }  0\le k \le n-1  \,.
\end{equation}

\noindent\myul{\textbf{Case 1:} $v_1(g_1) = \alpha_1$.} In this case, the initial vector $H'_0$ given as input for the sub-instance is defined by $H'_0[i]=H_0[i]-\bm{1}_{N_h(g_1, 1)}(i) \ge n-1$ for $i\in \{2, \ldots, n\}$, where $\bm{1}_{S}(i)$ is the indicator function of whether $i\in S$. Notice that this way, the priority among agents remains exactly the same as in the original instance and  $H'_{t'}[i] = H_{t'+1}[i]$ for all $t'\in\{0, \ldots, n-1\}$ and $i \in \{2, \ldots, n\}$; of course, $H'_{t'}[1]$ is not defined.
Further, because agent $1$ gets a high-valued good at time $1$, $H_1[1] = H_0[1] -1 + 3n -2 \ge 4n-3$ and, thus, throughout Phase 0 of the original instance (i.e., $t\in \{0, 1, \ldots, n\}$) we have $H_t[1] \ge 4n-3 - (t-1) \ge 3n -2 \ge n$. 
So, agent $1$ may only appear in $\mathcal{H}_n(t)$ for any $t\in [n]$. 
By the discussion about the correspondence between $H'$ and $H$ above, this means that $\mathcal{H}'_{\ell}(t') = \mathcal{H}_{\ell}(t'+1)$ for $\ell\in \{0, 1, \ldots, n-1\}$. Therefore, by the induction hypothesis, at the end of any $t\in \{1, \ldots, n\}$, $\big| \bigcup_{\ell=0}^{k} \mathcal{H}_{\ell}(t) \big| = \big| \bigcup_{\ell=0}^{k} \mathcal{H}'_{\ell}(t-1) \big| \le k$, for all $0\le k \le n-1$. For the missing cases, namely $t=0$ and $0\le k \le n$ or $0\le t \le n$ and $k = n$, we note that they are both trivial: (i) for $t = 0$, any $\mathcal{H}_{\ell}(0)$ is empty with the possible exception of $\mathcal{H}_n(0)$ that may contain up to $n$ agents, and (ii) for $k = n$, $\big| \bigcup_{\ell=0}^{n} \mathcal{H}_{\ell}(t) \big| \le n$ trivially holds for any $t$. We conclude that at the end of any $t\le n$, we have $\big| \bigcup_{\ell=0}^{k} \mathcal{H}_{\ell}(t) \big| \le k$, for all $0\le k \le n$.
\medskip 

\noindent\myul{\textbf{Case 2:} $v_1(g_1) = \beta_1 < \alpha_1$} Since the very first good, $g_1$, was allocated as a low-valued good despite $\chi_i = 0$, for all $i\in N$, it must be the case that $v_i(g_1) = \beta_i$. That is, $H_1[i]=H_0[i] \ge n$, for all $i\in N$.
In this case, the initial vector $H'_0$ of the sub-instance given as input is defined by $H'_0[i]=H_0[i] - 1 \ge n-1$ for $i\in \{2, \ldots, n\}$. Like in Case 1, the priority among agents remains exactly the same as in the original instance but here  $H'_{t'}[i] = H_{t'+1}[i] - 1$ for all $t'\in\{0, \ldots, n-1\}$ and $i \in \{2, \ldots, n\}$; again, $H'_{t'}[1]$ is not defined.

Unlike Case 1, however, here agent $1$ may  appear in $\mathcal{H}_{\ell}(t)$ for several combinations of $\ell$ and $t$. Nevertheless, this will not be an issue. First notice that, even if things went really wrong, there are not enough goods for $H'_{t'}[i]$ to become negative for any  $i\in \{2, \ldots, n\}$ and any $t'\in\{0, \ldots, n-1\}$, and so $\mathcal{H}'_{-1}(t') = \emptyset$ for all $t'\in\{0, \ldots, n-1\}$.
By the correspondence between $H'$ and $H$ discussed above, we have that $\mathcal{H}_{\ell}(t) \setminus \{1\} = \mathcal{H}'_{\ell-1}(t-1)$ and, thus, $\mathcal{H}_{\ell}(t) \subseteq \mathcal{H}'_{\ell-1}(t-1) \cup \{1\}$, for $\ell\in \{0, 1, \ldots, n\}$. Therefore, by the induction hypothesis, at the end of any $t\in \{1, \ldots, n\}$, 
\[\bigg| \bigcup_{\ell=0}^{k} \mathcal{H}_{\ell}(t) \bigg| \le \bigg| \bigcup_{\ell=0}^{k} \mathcal{H}'_{\ell-1}(t-1) \cup \{1\} \bigg| \le \bigg| \bigcup_{\ell=0}^{k-1} \mathcal{H}'_{\ell}(t-1) \bigg| + 1 \le k-1 +1 = k \,,\]
for all $0\le k \le n$. The missing cases ($t=0$ and $0\le k \le n$) are trivial exactly like their counterparts in Case 1. We conclude that at the end of any $t\le n$,  we have $\big| \bigcup_{\ell=0}^{k} \mathcal{H}_{\ell}(t) \big| \le k$, for all $0\le k \le n$. 
\end{proof}

It is not hard to see that the proof of Claim \ref{claim:H-condition_<=n} cannot be extended beyond $t=n$ as it crucially relies on the fact that during Phase $0$ (which has fixed length and lasts up to $t=n$) once an agent gets a good, they are inactive for the rest of the phase. Interestingly enough, the proof of Claim \ref{claim:H-condition_>n} below (induction with respect to $t$) could not have been used for $t<n$ as it crucially depends on $H[i]$ not being relevant unless agent $i$ actively competes for a high-valued good is allocated as high-valued (which, in turn, is the result of the ``slack'' we add to  $H[i]$ after an agent $i$ gets a high-valued good).

\begin{claim}\label{claim:H-condition_>n}
	At the end of any time step $t\ge n$, we have $\big| \bigcup_{\ell=0}^{k} \mathcal{H}_{\ell}(t) \big| \le k$, for all $0\le k \le n$.
\end{claim}

\begin{proof}[Proof of Claim \ref{claim:H-condition_>n}]
\renewcommand\qedsymbol{{\scriptsize \textbf{Cl.~\ref{claim:H-condition_>n}} $\boxdot$}}
Given an arbitrary instance, we will prove the statement using strong induction on the time step $t$. Essentially, Claim \ref{claim:H-condition_<=n} serves as the  base case.
So, assume that the statement of the claim is true for all time steps up to a certain $t_0\ge n$, and consider the next time step $t=t_0 +1$. Let $g$ be the $t$-th good and $j$ be the agent who eventually gets it. 
\medskip 

\noindent\myul{\textbf{Case 1:} $v_j(g) = \beta_j < \alpha_j$.}
That is, $g$ is allocated as a low-valued good. 
We claim that for any agent $i\in N$, such that $H_{t}[i]\le n-1$, we have $H_t[i] = H_{t-1}[i]$, i.e., the entries of $H_t$ may have changed only for agents who are irrelevant with respect to condition \eqref{eq:main_condition}. 
Indeed, if $H_t[i]$ has changed, then $v_i(g) = \alpha_i$. Moreover, it must be $\chi_i = 0$, as otherwise $N_h(g, t)\neq \emptyset$ and $g$ would be allocated as a high-valued good
instead. But $\chi_i = 0$ means that agent $i$ has received a high-valued good, say at a time step $t_h$, during the current phase. 
Recall that each phase has at most $2n-1$ time steps (see the proof of part 3.~of Theorem \ref{thm:main_algorithm}) and, thus,  $t-t_h \le 2n-2$. Also, by the induction hypothesis (for time step $t_h$ and $k=0$), we have $H_{t_h-1}[i]\ge 1$. Combining these, we get
\begin{equation}\label{eq:H_for_chi=0}
    H_{t}[i] \ge H_{t_h}[i] - (2n-2) = H_{t_h-1}[i] - 1 + 3n - 2 - 2n +2 \ge n  \,,
\end{equation}
as claimed. 
This means that $\mathcal{H}_{\ell}(t) = \mathcal{H}_{\ell}(t-1)$ for $\ell\in \{0, 1, \ldots, n-1\}$. Therefore, by invoking the induction hypothesis, at the end of time step $t$, we have  $\big| \bigcup_{\ell=0}^{k} \mathcal{H}_{\ell}(t) \big| = \big| \bigcup_{\ell=0}^{k} \mathcal{H}_{\ell}(t-1) \big| \le k$, for all $0\le k \le n-1$. For the missing case, namely $k = n$, it is trivial as $\big| \bigcup_{\ell=0}^{n} \mathcal{H}_{\ell}(t) \big| \le n$ always holds for any $t$. We conclude that at the end of time step $t$, we have $\big| \bigcup_{\ell=0}^{k} \mathcal{H}_{\ell}(t) \big| \le k$, for all $0\le k \le n$.
\medskip

\noindent\myul{\textbf{Case 2:} $v_j(g) = \alpha_j$.}
That is, we assume next that $g$ is allocated as a high-valued good.
Let $N_h^+(g, t) = \{i\in N \,|\, v_i(g) = \alpha_i\}$, i.e., the set of all agents who see the $t$-th good as high-valued. The agents of $N_h^+(g, t)$ are exactly those whose entries in $H_t$ are updated during the current time step and how relevant it is for our analysis.  
We carefully categorize agents according to what happens to their entry in $H_t$ during time step $t$:
\begin{itemize}[leftmargin=20pt,itemsep=3pt]
    \item For any agent $i\in N\setminus N_h^+(g, t)$, $H_t[i] = H_{t-1}[i]$, as they see $g$ as a low-valued good. 
    \item For any agent $i\in  N_h(g, t) \setminus \{j\}$, $H_t[i] = H_{t-1}[i] - 1$, as they see $g$ as a high-valued good and they miss it. 
    \item For any agent $i\in  N_h^+(g, t) \setminus N_h(g, t)$, again $H_t[i] = H_{t-1}[i] - 1$, as they see $g$ as a high-valued good and they miss it, but  they are essentially irrelevant because  $H_{t}[i]\ge n$. This follows by the argument in Case 1 above, as $\chi_i = 0$ and the chain of (in)equalities of \eqref{eq:H_for_chi=0} applies exactly as is.
    \item For agent $j$ itself, we have $H_t[j] = H_{t-1}[j] - 1 +3n -2 \ge 3n - 2\ge n$, where $H_{t-1}[j]\ge 1$ follows from the induction hypothesis for time step $t-1$ and $k=0$. We conclude that agent $j$ is also essentially irrelevant.   
\end{itemize}
From the above, it becomes clear that agents in $N_h(g, t) \cap \bigcup_{\ell=0}^{k} \mathcal{H}_{\ell}(t)$ should be treated carefully when we try to bound $\big| \bigcup_{\ell=0}^{k} \mathcal{H}_{\ell}(t) \big|$, for some $k\in \{0, 1, \ldots, n\}$. The easiest case, of course, is when $N_h(g, t) \cap \bigcup_{\ell=0}^{k} \mathcal{H}_{\ell}(t) = \emptyset$. Then, $\mathcal{H}_{\ell}(t) = \mathcal{H}_{\ell}(t-1)$ for $\ell \le k$ and, thus, by the induction hypothesis, we have that at the end of time step $t$, $\big| \bigcup_{\ell=0}^{k} \mathcal{H}_{\ell}(t) \big| = \big| \bigcup_{\ell=0}^{k} \mathcal{H}_{\ell}(t-1) \big| \le k$. Also, when $k=n$, it trivially holds $\big| \bigcup_{\ell=0}^{n} \mathcal{H}_{\ell}(t) \big| \le n$.

Next, assume that $k\in \{0, 1, \ldots, n-1\}$ is such that $N_h(g, t) \cap \bigcup_{\ell=0}^{k} \mathcal{H}_{\ell}(t) = S \neq\emptyset$. At this point, we need three simple observations. The first one (following from the last bullet above) is that $j\notin \bigcup_{\ell=0}^{k} \mathcal{H}_{\ell}(t)$ and, thus, $j\notin S$. The second one is that $\bigcup_{\ell=0}^{k} \mathcal{H}_{\ell}(t) \subseteq \bigcup_{\ell=0}^{k+1} \mathcal{H}_{\ell}(t-1)$, as no entry of $H$ reduces by more than $1$ in a single time step. The last one is that $j \in \bigcup_{\ell=0}^{k+1} \mathcal{H}_{\ell}(t-1)$; if not, we would have $H_{t-1}[j]\ge k+2$ and the $j$-th entry of $H$ right before $g$ was allocated to $j$ would be $H_{t-1}[j] - 1\ge k+1 > k \ge \min_{i\in S} H[i] \ge \min_{i\in N_h(g, t)} H[i]$, contradicting the choice of $j$.
Using these observations, as well as the induction hypothesis,  we have that at the end of time step $t$,  
\[\bigg| \bigcup_{\ell=0}^{k} \mathcal{H}_{\ell}(t) \bigg| \le \bigg| \bigcup_{\ell=0}^{k+1} \mathcal{H}_{\ell}(t-1) \setminus \{j\} \bigg| = \bigg| \bigcup_{\ell=0}^{k+1} \mathcal{H}_{\ell}(t-1)\bigg| - 1 \le (k+1) -1 = k. \]
This exhausts all possible cases for $k\in \{0, 1, \ldots, n\}$ and concludes Case 2.
\end{proof}

Clearly, combining the two claims completes the proof of the lemma.
\end{proof}

%%%%%%%%%%%%%%%%%%%%%%%%%%

%%%%%%%%%%%%%%%%%%%%%%%%%%%%%%%%%%%%%%%%%
%%%%%%%%%%%%%%%%%%%%%%%%%%%%%%%%%%%%%%%%%
\section{The Power of Limited Foresight}
\label{sec:foresight_whole}
%%%%%%%%%%%%%%%%%%%%%%%%%%%%%%%%%%%%%%%%%
%%%%%%%%%%%%%%%%%%%%%%%%%%%%%%%%%%%%%%%%%

A natural question at this point is whether one can avoid the linear term in the \mms approximation guarantee by allowing some additional information about the future. Unfortunately, there is a very simple instance that illustrates that this is not possible. 

\begin{proposition}\label{prop:linear_impossibility}
Let $\varepsilon > 0$ be a constant. Even if the whole instance is known in advance, as long as it is required to irrevocably allocate each good right after it arrives, no algorithm can always compute a $(1/n + \varepsilon)$-temporal-\mms allocation, even on $2$-value instances. 
\end{proposition}

\begin{proof}
We consider a $2$-value instance with $n$ agents, where $\beta_i = 1$ and $\alpha_i = \alpha \ge n$, for all $i\in N$. There are $n$ universally low-valued goods $g_1, \ldots, g_{n}$ and $n-1$ universally high-valued goods $g_{n+1}, \ldots, g_{2n-1}$, as shown below, that arrive according to their indices. 
\begin{center}  
\begin{tabular}{lccccc|cccc}
 & $g_1$ & $g_2$ &  \ldots & $g_{n-1}$ & $g_n$ & $g_{n+1}$ & $g_{n+2}$ & \ldots & $g_{2n-1}$  \\
ag.~1:  & \circled{1} & 1 &   \ldots & 1 & 1 & $\alpha$ & $\alpha$ & \ldots & $\alpha$  \\
ag.~2:  & 1 & \circled{1} &   \ldots & 1 & 1 & $\alpha$ & $\alpha$ & \ldots & $\alpha$  \\
$\vdots$  & $\vdots$ & $\vdots$ &    $\ddots$ & $\vdots$  & $\vdots$ & $\vdots$  & $\vdots$ & $\ddots$  & $\vdots$ \\
ag.~$n-1$:  & 1 & 1 & \ldots & \circled{1} & 1 & $\alpha$ & $\alpha$ & \ldots & $\alpha$  \\
ag.~$n$:  & 1 & 1 &  \ldots & 1 & \circled{1} & $\alpha$ & $\alpha$ & \ldots & $\alpha$  \\[2pt]
\end{tabular}
\end{center}
We may assume that this information is available even before the first good arrives.
Like in the proof of Theorem \ref{thm:imposibility_MMS}, we note that if at any point during the first $n$ time steps an agent receives a second good, then at the end of time step $t=n$ there would be at least one agent that has received value $0$ despite having a positive maximin share value. So, we may assume that any algorithm with non-trivial guarantees assigns to each agent exactly one of the first $n$ goods. Without loss of generality, agent $i$ gets good $g_i$ as shown. Given that, no matter how goods $g_{n+1}, \ldots, g_{2n-1}$ are allocated, at least one agent, say agent $1$, will not receive any of these. So, at the end of time step $t = 2n-1$ we have $v_i(A_1) = v_1(g_1) = 1$, while $\bmu_1(2n-1) = n$.
\end{proof}

Despite Proposition \ref{prop:linear_impossibility}, in this section we show that being able to see only a linear number of steps in the future leads to significantly simpler algorithms with \efo and \eft guarantees.

%%%%%%%%%%%%%%%%%%%%%%%%%%%%%%%%%%%%%%%%%
\subsection{The Illustrative Case of Two Agents}
\label{sec:foresight_2} 
%%%%%%%%%%%%%%%%%%%%%%%%%%%%%%%%%%%%%%%%%

There is an easy algorithm augmented with foresight of length $1$  that achieves \efo in every \textit{even} time step. Although Naive-Matching (Algorithm \ref{alg:foresight_2}) is essentially subsumed by the main result of the next section, it is simpler to state, much simpler to analyze and still illustrates the power of---even very limited---foresight.

\begin{table}[h!t]
\begin{center}
\small
\begin{tabular}{ccccccccccc}
\circled{$\beta$} & $\beta$                      & \hspace{1.3em}  & $\beta$                      & \circled{$\alpha$}                      & \hspace{1.3em}  & \circled{$\alpha$}            &  $\beta$          & \hspace{1.3em} &  \circled{$\alpha$}            &  $\alpha$ \\
$\beta$ & \circled{$\beta$}                      &  & \circled{$\beta$}                      & $\beta$                      &  &       $\beta$ & \circled{$\beta$}       &  &  $\beta$ & \circled{$\beta$}  \\
 & \vspace{1.4ex} &                       &  &                       &                       &  &                       &  &  &  \\ \cline{4-5}
\circled{$\beta$} & $\beta$  & \multicolumn{1}{r|}{\textcolor{blue}{$\mathbf{I}$}} & \dotcircled{$\beta$} & \multicolumn{1}{c|}{\dashcircled{$\alpha$}} &                       & \circled{$\alpha$}            &  $\beta$          & \hspace{1.3em} &  \circled{$\alpha$}            &  $\alpha$ \\
$\beta$ & \circled{$\alpha$}  & \multicolumn{1}{l|}{} &  \dashcircled{$\beta$} &  \multicolumn{1}{c|}{\dotcircled{$\alpha$}} &                       &  $\beta$ & \circled{$\alpha$}                       &  &  $\beta$ & \circled{$\alpha$}  \\ \cline{4-5}
 & \vspace{1.4ex} &                       &  &                       &                       &  &                       &  &  &  \\ \cline{7-8}
$\beta$ & \circled{$\beta$} &                       & $\beta$ &  \circled{$\alpha$}                     & \multicolumn{1}{r|}{\textcolor{blue}{$\mathbf{II}$}} & \dashcircled{$\alpha$} & \multicolumn{1}{c|}{\dotcircled{$\beta$}} &  & $\alpha$ & \circled{$\alpha$} \\
\circled{$\alpha$} & $\beta$ &                       & \circled{$\alpha$} &  $\beta$                     & \multicolumn{1}{c|}{} & \dotcircled{$\alpha$} & \multicolumn{1}{c|}{\dashcircled{$\beta$}} &  & \circled{$\alpha$} & $\beta$ \\ \cline{7-8}
 & \vspace{1.4ex} &                       &  &                       &                       &  &                       &  &  &  \\
$\beta$               & \circled{$\beta$} &    & $\beta$                      & \circled{$\alpha$}                      &   & \circled{$\alpha$}            &  $\beta$          &  &  \circled{$\alpha$}            &  $\alpha$ \\
\circled{$\alpha$}         & $\alpha$ &   & \circled{$\alpha$}                      & $\alpha$                      &  &       $\alpha$ & \circled{$\alpha$}       &  &  $\alpha$ & \circled{$\alpha$} 
\end{tabular}
\normalfont
\end{center}
\caption{When we are given foresight of length $1$, achieving \efo on every other time step for two agents is very simple. We only need to keep track of who ``wins'' the next block of type I or II, i.e., who gets the contested high-valued good (\textit{dashed:} agent 1 wins, \textit{dotted:} agent 2 wins). In every other case, we may allocate the goods in a predetermined way as shown for any other block here. Since we only care about the general value patterns, for the sake of readability, we omit the indices; e.g., we write $\alpha$ rather than $\alpha_1$ in the first row and $\alpha_2$ in the second row of each block.}
\label{table:foresight_2}
\end{table}

\begin{algorithm}[bht]
		\caption{Naive-Matching$(v_1, v_2; M)$ \\{\small {(The valuation functions, $v_1, v_2$, are given via oracles; $M$ is given in an online fashion, one good at a time, along with a preview of the next good after that.)}}}
		\begin{algorithmic}[1]
            \State $t\leftarrow 0$ \textbf{;} $\mathtt{ctr} \leftarrow 0$ \Comment{We initialize our counters and the allocation.}
			\For{$i \in N$}
                \State $A_i \leftarrow \emptyset$ \vspace{1ex} 
            \EndFor\vspace{1ex}
            \State \textbf{whenever a new good $g$ arrives along with a preview of the next good, $g'$:}
            \State $t\leftarrow t + 1$
            \If{$t = 0 \,\mathrm{mod}\,2$}
                \State Allocate $g$ according to the commitment of time step $t-1$.
            \Else \label{line:naive_else}
                \State $B \leftarrow {\small \begin{bmatrix} v_1(g) & v_1(g') \\ v_2(g) & v_2(g') \end{bmatrix}}$
                \If{$B$ follows any pattern of Table \ref{table:foresight_2} except of those of blocks I or II \label{line:naive_else->if}}
                    \State Add $g$ to $A_1$ or $A_2$ according to the allocation of the corresponding pattern in Table \ref{table:foresight_2}.
                    \State Commit to add $g'$ to $A_1$ or $A_2$ at time step $t+1$ according to the aforementioned allocation.
                \Else
                    \State $\mathtt{ctr} \leftarrow (\mathtt{ctr} + 1)\,\mathrm{mod}\,2$ \label{line:naive_ctr_update}
                    \State $j = 2 - \mathtt{ctr}$
                    \State Add $g$ to $A_1$ or $A_2$ according to the allocation of the corresponding pattern in Table \ref{table:foresight_2} (i.e., I or II) which gives the contested high-valued good to agent $j$.
                    \State Commit to add $g'$ to $A_1$ or $A_2$ at time step $t+1$ according to the aforementioned allocation. \label{line:naive_I-II_commit}
                \EndIf
            \EndIf
		\end{algorithmic}
		\label{alg:foresight_2}
\end{algorithm}

\begin{theorem}\label{thm:foresight_2}
For any two agent $2$-value instance augmented with foresight of length $1$, Algorithm \ref{alg:foresight_2} builds an allocation that is temporal-\eft, while it is \efo for every \emph{even} time step $t\ge 0$. 
\end{theorem}

\begin{proof}
We first observe that if an allocation is \efo at the end of some time step $t\ge 0$, then it will trivially be \eft at the end of  time step $t+1$, no matter how the corresponding good is allocated. That is, it suffices to show that Algorithm \ref{alg:foresight_2} builds an allocation that is \efo for every \emph{even} time step $t$, and the fact that it is also temporal-\eft immediately follows. We will show  a somewhat stronger statement for even time steps: \emph{for any $k\in \mathbb{Z}_{\ge 0}$, at the end of time step $t = 2k \le m$, \emph{(i)} both bundles contain $k$ goods each, and \emph{(ii)} either $\mathtt{ctr} = 0$ and only agent $1$ may envy agent $2$ by at most $\alpha_1 - \beta_1$, or $\mathtt{ctr} = 1$ and only agent $2$ may envy agent $1$ by at most $\alpha_2 - \beta_2$.} Of course, (i) is straightforward because, for any $k \ge 0$, we always create a matching between the two agents and the goods $g_{2k + 1}, g_{2k + 2}$. We are going to show (ii) using induction on $k$.

At time $t = 0$ (i.e., for $k = 0$) the statement of (ii) trivially holds: $\mathtt{ctr} = 0$ and no agent envies the other.  Now assume that (ii) holds for some $t = 2k$, such that $k\ge 0$ and $t = 2k +2 \le m$. At time $t = 2k +1$ the algorithm enters the `else' in line \ref{line:naive_else}. Note that whenever the condition of line \ref{line:naive_else->if} is true, i.e., goods $g_{2k + 1}, g_{2k + 2}$ induce any pattern of Table \ref{table:foresight_2} except of those of blocks I or II, each agent (weakly) prefers the good they receive according to the corresponding allocation shown in Table \ref{table:foresight_2}. That is, for agent $1$ we have $v_1(A_1^{2k + 2})-v_1(A_1^{2k}) \ge v_1(A_2^{2k + 2})-v_1(A_2^{2k})$, where the added superscript indicates the time step at the end of which we consider the bundle, and similarly for agent $2$. This observation that envy can only be reduced in this case, combined with the fact that $\mathtt{ctr}$ does not change here, implies that (ii) still holds for $t = 2k +2$. 

So, we may assume that the condition of line \ref{line:naive_else->if} is false, i.e., goods $g_{2k + 1}, g_{2k + 2}$ induce the pattern of  block I or block II of Table \ref{table:foresight_2} with one universally high-valued and one universally low-valued good. Here we consider the two cases of the induction hypothesis.
\medskip

\noindent\underline{\textbf{Case 1:} at $t = 2k$, $\mathtt{ctr} = 0$, $v_1(A_1^{2k}) \ge v_1(A_2^{2k}) - (\alpha_1 - \beta_1)$ and $v_2(A_2^{2k}) \ge v_2(A_1^{2k})$.}
According to lines \ref{line:naive_ctr_update}-\ref{line:naive_I-II_commit}, in this case, $\mathtt{ctr}$ becomes $1$ and the algorithm commits to giving the high-valued good to agent $1$. Thus,  
\[v_1(A_1^{2k+2})  = v_1(A_1^{2k}) + \alpha_1 \ge v_1(A_2^{2k}) - (\alpha_1 - \beta_1) + \alpha_1 =  v_1(A_2^{2k+2})\,,\]
as well as
\[v_2(A_2^{2k+2})  = v_2(A_2^{2k}) + \beta_2 \ge v_2(A_1^{2k}) + \beta_2 =  v_2(A_1^{2k+2}) - \alpha_2 + \beta_2\,,\]
i.e., (ii) still holds for $t = 2k +2$.
\medskip

\noindent\underline{\textbf{Case 2:} at $t = 2k$, $\mathtt{ctr} = 1$, $v_1(A_1^{2k}) \ge v_1(A_2^{2k})$ and $v_2(A_2^{2k}) \ge v_2(A_1^{2k})  - (\alpha_2 - \beta_2)$.} This is completely symmetric to Case 1 above. According to lines \ref{line:naive_ctr_update}-\ref{line:naive_I-II_commit}, $\mathtt{ctr}$ becomes $0$ and the algorithm commits to giving the high-valued good to agent $2$. Thus,  
\[v_1(A_1^{2k+2})  = v_1(A_1^{2k}) + \beta_1 \ge v_1(A_2^{2k}) + \beta_1 =  v_1(A_2^{2k+2}) - \alpha_1 + \beta_1\,,\]
and
\[v_2(A_2^{2k+2})  = v_2(A_2^{2k}) + \alpha_2 \ge v_2(A_1^{2k}) - (\alpha_2 - \beta_2) + \alpha_2 =  v_2(A_1^{2k+2})\,,\]
i.e., (ii) still holds for $t = 2k +2$.
\medskip

This concludes the induction and shows that Algorithm \ref{alg:foresight_2} builds an allocation that is \efo for every even time step $t$ and, thus, temporal-\eft.
\end{proof}

\begin{corollary}\label{cor:foresight_2-asymptotics}
    Assuming $m$ is large enough, for any $\lambda \in \mathbb{Z}_{>0}$, after a sufficient number of steps, the allocation built by Algorithm \ref{alg:foresight_2} becomes and remains $\lambda/(\lambda+2)$-\ef, $\lambda/(\lambda+1)$-\efo, and $\lambda/(\lambda+1)$-\prop. 
\end{corollary}

\begin{proof}
Clearly, if $m$ is large enough (e.g., $m\ge 2 \lambda \max_{i\in[2]}\lceil\alpha_i \rceil$) there is some $t_*$ by which each agent $i$ has received value equal to at least $\lambda \alpha_i$. At the end of any time step $t\ge t_*$, we have for agent $1$
\[v_1(A_2^{t}) \le \min_{g, g'\in A_2^{t}} v_1(A_2^{t}\setminus\{g, g'\}) + 2 \alpha_1 \le  v_1(A_1^{t}) + 2 \alpha_1 \le (1 + 2/\lambda) v_1(A_1^{t})\,, \]
\[\min_{g\in A_2^{t}} v_1(A_2^{t}\setminus\{g\}) \le \min_{g, g'\in A_2^{t}} v_1(A_2^{t}\setminus\{g, g'\}) + \alpha_1 \le  v_1(A_1^{t}) +  \alpha_1 \le (1 + 1/\lambda) v_1(A_1^{t})\,, \text{\ \ and\ \ }\]
\[0.5(v_1(A_1^{t}) + v_1(A_2^{t}))  \le 0.5\big[v_1(A_1^{t}) + (1 + 2/\lambda)v_1(A_1^{t})\big] = (1 + 1/\lambda) v_1(A_1^{t})\,, \]
and similarly for agent $2$.
\end{proof}

%%%%%%%%%%%%%%%%%%%%%%%%%%%%%%%%%%%%%%%%%
\subsection{Foresight of Length $n-1$ Suffices}
\label{sec:foresight_n}
%%%%%%%%%%%%%%%%%%%%%%%%%%%%%%%%%%%%%%%%%
As we mentioned in the beginning of the previous section, we can generalize Theorem \ref{thm:foresight_2} to any number of agents, albeit with a more complicated algorithm. 

Similarly to Algorithm \ref{alg:foresight_2}, what we would like to achieve between time steps $kn$ and $(k+1)n$ (i.e., with goods $g_{kn+1}, \ldots, g_{(k+1)n}$) is to obtain an \efo allocation by taking the union of the \efo allocation of time step $kn$ and an appropriately chosen matching. However, it is easy to see that this fails if not done carefully, even for two agents, which is the reason why Algorithm \ref{alg:foresight_2} treats the patterns of blocks I and II with some care. In order to achieve a similar thing for general $n$, at time step $kn+1$ we construct a matching $\mathcal{M}$ involving the current good and the $n-1$ predicted goods, so that the matching's envy graph ``cancels out'' any problematic edges of the envy graph  of the \efo allocation of time step $kn$. Then, for what we call the $(k+1)$-th round, we allocate these $n$ goods according to $\mathcal{M}$. For $\mathcal{M}$ to have the nice aforementioned behavior, however, we take it to be a maximum weight matching with respect to carefully defined auxiliary valuation functions. These auxiliary functions encode all the information about which goods are high or low for which agents, while giving additional weight to goods for agents who come earlier in the topological sorting. The latter captures the intuition that such agents should have a higher priority during this round.

\begin{algorithm}[h!t]
		\caption{Priority-Matching$(v_1, \ldots, v_n; M)$ \\{\small {(The valuation functions, $v_i, i\in [n]$, are given via oracles; $M$ is given in an online fashion, one good at a time, along with a preview of the next $n-1$ goods after that.)}}}
		\begin{algorithmic}[1]
            \State $t\leftarrow 0$ \textbf{;} $\mathtt{ctr} \leftarrow 0$ \Comment{We initialize our counters and the allocation.}
			\For{$i \in N$}
                \State $A_i \leftarrow \emptyset$ \vspace{1ex} 
            \EndFor\vspace{1ex}
            \State \textbf{whenever a new good $g$ arrives along with a preview of the next $n-1$ goods} 
            \State $t\leftarrow t + 1$
            
            \If{$t = 1 \,\mathrm{mod}\,n$} \label{line:priority_if} \Comment{A new \emph{round} begins, namely round number $\lceil t/n \rceil$.}
                \State Construct the current envy-graph $G_t = (V_t, E_t)$. \label{line:priority_envy_graph} \Comment{$V_t = N$ and $(i,j)\in E_t$ if $v_i(A_j)> v_i(A_i)$. }
                \State Find a permutation $\pi$ that induces a topological sorting of $G_t$. \label{line:priority_topological} \Comment{If $(i,j)\in E_t$, then $\pi(i) < \pi(j)$. }
                \For{$i\in [n]$ and $h\in \{g_{t}, g_{t+1}, \ldots, g_{t+n-1}\}$ } \Comment{Define auxiliary valuation functions consistent with $\pi$.}
                        \State $\tilde{v}_{\pi^{-1}(i)}(h)  \leftarrow {\small \begin{cases} 2 (1+ 1/2n)^{n-i} & \text{if }v_{\pi^{-1}(i)}(h) = \alpha_{\pi^{-1}(i)} \\ (1+ 1/2n)^{n-i} & \text{otherwise} \end{cases}}$ \label{line:priority_aux} \Comment{$\pi^{-1}(i)$ is the $i$-th agent in the ordering.}
                \EndFor\vspace{0.6ex}
                \State Find a maximum weight  matching $\mathcal{M}$ between the agents and goods in $\{g_{t}, \ldots, g_{t+n-1}\}$ with respect to these auxiliary functions. \Comment{The weight of a pair $(j, h)$ is $\tilde{v}_{j}(h)$. The matching is perfect unless this is the last round; in this case, the matching is between $N$ and $\{g_{t}, \ldots, g_{m}\}$.}
                \State Add $g = g_t$ to a bundle according to $\mathcal{M}$, i.e., add $g_t$ to $A_j$ if and only if $(j, g_t)\in \mathcal{M}$.
                \State Commit to also allocate $\{g_{t+1}, \ldots, g_{t+n-1}\}$ according to $\mathcal{M}$ at time steps $t+1$ through $t+n-1$.
            \Else \label{line:priority_else} \Comment{The allocation of $g$ was decided at the beginning of this round.} %, $k-1$ time steps ago.}
                \State Allocate $g$ according to the commitment of time step 
                $(\lceil t/n \rceil - 1)n +1$ (when this round begun). 
            \EndIf
		\end{algorithmic}
		\label{alg:foresight_n}
\end{algorithm}

\begin{theorem}\label{thm:foresight_n-b}
For any personalized $2$-value instance augmented with foresight of length $n-1$, Algorithm \ref{alg:foresight_n} builds an allocation that is temporal-\eft, while it is \efo (and, thus, $1/n$-\mms) for every time step $t = kn, k\in \mathbb{Z}_{\ge 0}$. Moreover, if at any step $t_{0}$ the allocation fails to be $1/2$-\efo, then it remains $1/2$-\efo  at the end of every time step $t\ge \lceil t_0 / n \rceil n$. 
\end{theorem}

\begin{proof}
The proof has similar structure to the proof of Theorem \ref{thm:foresight_2}, but the induction is fairly more complicated due to the fact that the matching process is much less trivial here. In what follows, we refer to the execution of Algorithm \ref{alg:foresight_n} for time steps $kn+1, kn+1, \ldots, (k+1)n$ as the $(k+1)$-th \emph{round}, for any $k\in \mathbb{Z}_{\ge 0}$.

We begin with a generalization of the first observation made in the 
proof of Theorem \ref{thm:foresight_2}:
if an allocation is \efo at the end of some time step $t\ge 0$, then it will trivially be \eft at the end of time step $t+\ell$, as long as no bundle receives more than one out of the last $\ell$ goods. For Algorithm \ref{alg:foresight_n}, this means that if the allocation is \efo at the end of time step $t = kn$ (i.e., right before round $k+1$ begins), for every $k\in \mathbb{Z}_{\ge 0}$, then it remains \eft for every time step in between. This holds just because in every round each agent receives (at most) one good via the corresponding matching $\mathcal{M}$. So, it suffices to show that Algorithm \ref{alg:foresight_n} builds an allocation that is \efo for every time step $t$ which is a multiple of $n$, and the fact that it is also temporal-\eft immediately follows. Again we show a somewhat stronger statement: \emph{for any $k\in \mathbb{Z}_{\ge 0}$, at the end of time step $t = kn \le m$,
\begin{enumerate}[leftmargin=25pt,itemsep=3pt]
    \item[(i)] all bundles contain $k$ goods each;
    \item[(ii)] the corresponding envy graph (i.e., $G_{kn+1}$ constructed in the beginning of the next time step in line \ref{line:priority_envy_graph}) is acyclic;
    \item[(iii)] any edge $(i, j)$ in the above graph indicates envy 
    which is at most $\alpha_i - \beta_i$.
\end{enumerate}}
Like before, (i) is straightforward, since the algorithm always creates a matching between the agents and (at most) $n$ goods in the beginning of a round and, as these goods arrive, it allocates them according to the matching. We are going to show (ii) and (iii) using induction on $k$.
At time $t = 0$ (i.e., for $k = 0$) the statements of (ii) and (iii) trivially hold, as no agent envies another agent.  Now assume that (ii) and (iii) hold at the end of some time step $t = kn$, such that $k\ge 0$ and $t = (k + 1)n\le m$.  At the beginning of time step $t = kn +1$ the algorithm enters the `if' in line \ref{line:priority_if} and constructs the envy graph $G:= G_{kn+1}$ in line \ref{line:priority_envy_graph}. This $G$ is exactly the envy graph for which the two parts of the induction hypothesis hold. The fact that $G$ is acyclic is what allows it to be topologically sorted (see, e.g., \citet{cormen2022introduction}) and, hence, makes line \ref{line:priority_topological} well-defined. 
We use $G'$ to denote the envy graph at the end of the current (i.e., the $(k+1)$-th) round; note that $G'$ is constructed as $G_{(k+1)n+1}$ at the beginning of the next round.

\begin{claim}\label{claim:priority_envy}
Any edge $(i, j)$ of $G'$ indicates envy which is at most $\alpha_i - \beta_i$.
\end{claim}

\begin{proof}[Proof of Claim \ref{claim:priority_envy}]
\renewcommand\qedsymbol{{\scriptsize \textbf{Cl.~\ref{claim:priority_envy}} $\boxdot$}}
Consider any edge $(i, j)$ of $G'$ and let $h_i, h_j$ be the goods agents $i$ and $j$ received, respectively, in round $k+1$. This means that $h_i, h_j$ were matched with $i$ and $j$, respectively, in $\mathcal{M}$. By the induction hypothesis, we know that the envy from agent $i$ towards agent $j$ at the end of time step $t = kn$---if it existed at all---was upper bounded by $\alpha_i - \beta_i$, i.e., $v_i(A_j^{kn}) - v_i(A_i^{kn}) \le \alpha_i - \beta_i$. If $v_i(h_i) \ge v_i(h_j)$, then  
\[v_i(A_j^{(k+1)n}) - v_i(A_i^{(k+1)n}) = v_i(A_j^{kn}) + v_i(h_j) - v_i(A_i^{kn}) - v_i(h_i) \le \alpha_i - \beta_i\,.\]

We need to consider the case where $\beta_i = v_i(h_i) < v_i(h_j) = \alpha_i$. If there was no  edge $(i, j)$ in $G$, i.e., if  $v_i(A_j^{kn}) - v_i(A_i^{kn}) \le 0$, then clearly
\[v_i(A_j^{(k+1)n}) - v_i(A_i^{(k+1)n}) = v_i(A_j^{kn}) + \alpha_i - v_i(A_i^{kn}) - \beta_i \le \alpha_i - \beta_i\,.\]
So we may assume that $(i, j)$ was an edge in $G$. Dy definition, this means that in any topological sorting agent $i$ must come before agent $j$. In particular, $\hat{\imath} = \pi(i) < \pi(j) = \hat{\jmath}$, i.e., in the sorting induced by the permutation $\pi$, agent $i$ is the $\hat{\imath}$-th agent and agent $j$ is the $\hat{\jmath}$-th. By the definition of the auxiliary functions in line \ref{line:priority_aux}, we have
\[\tilde{v}_{i}(h_{j}) = 2\Big(1 + \frac{1}{2n}\Big)^{n-\hat{\imath}} \text{\quad and \quad}
\tilde{v}_{i}(h_{i}) =  \Big(1 + \frac{1}{2n}\Big)^{n-\hat{\imath}}\,.\]
Let $\mathcal{M}'$ be the matching that one gets from $\mathcal{M}$ by switching  $h_i$ and $h_j$. That is, in $\mathcal{M}'$  agent $i$ is matched with  $h_j$, agent $j$ is matched with $h_i$, and every other agent $\ell$ is matched with $h_{\ell}$ as in $\mathcal{M}$. Now, if we use $w(\mathcal{M})$ to denote the sum of weights of the pairs in $\mathcal{M}$, we have 
\begin{IEEEeqnarray*}{rCl}
w(\mathcal{M}')-w(\mathcal{M}) & = & \Big(\tilde{v}_{i}(h_{j}) + \tilde{v}_{j}(h_{i})+ \!\!\!\sum_{\ell\in N\setminus\{i,j\}} \!\!\!\tilde{v}_{\ell}(h_{\ell}) \Big) - \sum_{\ell\in N} \tilde{v}_{\ell}(h_{\ell})
\\
& = & \tilde{v}_{i}(h_{j}) - \tilde{v}_{i}(h_{i})  + \tilde{v}_{j}(h_{i}) - \tilde{v}_{j}(h_{j})
\\
& \ge & 2\Big(1 + \frac{1}{2n}\Big)^{n-\hat{\imath}} - \Big(1 + \frac{1}{2n}\Big)^{n-\hat{\imath}} + \Big(1 + \frac{1}{2n}\Big)^{n-\hat{\jmath}} - 2\Big(1 + \frac{1}{2n}\Big)^{n-\hat{\jmath}} \\
& = & \Big(1 + \frac{1}{2n}\Big)^{n-\hat{\imath}} - \Big(1 + \frac{1}{2n}\Big)^{n-\hat{\jmath}} = \Big(1 + \frac{1}{2n}\Big)^{n-\hat{\jmath}} \Big[\Big(1 + \frac{1}{2n}\Big)^{\hat{\jmath}-\hat{\imath}} - 1 \Big] \\
& \ge &  \Big(1 + \frac{1}{2n}\Big)^{0} \Big[\Big(1 + \frac{1}{2n}\Big)^{1} - 1 \Big] = \frac{1}{2n} > 0 \,,
\end{IEEEeqnarray*}
where for the first inequality we used the exact values of $\tilde{v}_{i}(h_{j}), \tilde{v}_{i}(h_{i})$ and lower and upper bounds for  $\tilde{v}_{j}(h_{i}), \tilde{v}_{j}(h_{j})$. This, however, contradicts the choice of $\mathcal{M}$ as a maximum weight matching. Thus, under the assumption that $\beta_i = v_i(h_i) < v_i(h_j) = \alpha_i$, $(i, j)$ cannot be an edge in $G$. We conclude that, in any case, an edge $(i, j)$ of $G'$ indicates envy which is no more than $\alpha_i - \beta_i$.
\end{proof}

\begin{claim}\label{claim:priority_cycles}
The graph $G'$ is acyclic.
\end{claim}

\begin{proof}[Proof of Claim \ref{claim:priority_cycles}]
\renewcommand\qedsymbol{{\scriptsize \textbf{Cl.~\ref{claim:priority_cycles}} $\boxdot$}}
Suppose, towards a contradiction, that $G'$ contains a simple directed cycle $C = (i_1, i_2, \allowbreak \ldots, \allowbreak i_s, \allowbreak i_1)$. Also, let $h_{i_1}, \ldots, h_{i_s}$ denote the goods these agents received, respectively, in round $k+1$. Since $G$ was acyclic, not all edges of $C$ existed in $G$. Without loss of generality (as it is a matter of renaming the agents / vertices), we may assume that the $(i_1, i_2)$ was not an edge in $G$. Our first observation is that the only way this could happened, is that $v_{i_1}(h_{i_1}) = \beta_{i_1}$ but $v_{i_1}(h_{i_2}) = \alpha_{i_1}$. 

We next note that it must be the case that $v_{i_2}(h_{i_2}) = \alpha_{i_2}$, otherwise we could define a new matching $\mathcal{M}'$ by only switching $h_{i_1}$ and $h_{i_2}$ in $\mathcal{M}$ and improve the maximum weight, similarly to what we did in the proof of Claim \ref{claim:priority_envy}:
\begin{IEEEeqnarray*}{rCl}
w(\mathcal{M}')-w(\mathcal{M}) & = & \tilde{v}_{i_1}(h_{i_2}) - \tilde{v}_{i_1}(h_{i_1})  + \tilde{v}_{i_2}(h_{i_1}) - \tilde{v}_{i_2}(h_{i_2})
\\
& \ge & 2\Big(1 + \frac{1}{2n}\Big)^{n-\pi(i)} - \Big(1 + \frac{1}{2n}\Big)^{n-\pi(i)} + \Big(1 + \frac{1}{2n}\Big)^{n-\pi(j)} - \Big(1 + \frac{1}{2n}\Big)^{n-\pi(j)} \\
& \ge &  \Big(1 + \frac{1}{2n}\Big)^{0}  > 0 \,,
\end{IEEEeqnarray*}
where we used $v_{i_2}(h_{i_2}) = \beta_{i_2}$ for the first inequality, contradicting the choice of $\mathcal{M}$. So, it must be $v_{i_2}(h_{i_2}) = \alpha_{i_2}$.

The third observation we need for the proof is that for any edge $(i_r, i_{r+1})$ in $C$, such that $v_{i_r}(h_{i_r}) = \alpha_{i_r}$, it must also be the case that $v_{i_r}(h_{i_{r+1}}) = \alpha_{i_r}$.\footnote{\ We use the standard convention that $i_{s+1} := i_1$.} 
To see this, notice that 
\begin{IEEEeqnarray*}{rCl}
v_{i_r}(A_{i_{r+1}}^{(k+1)n}) - v_{i_r}(A_{i_r}^{(k+1)n}) & = & v_{i_r}(A_{i_{r+1}}^{kn}) - v_{i_r}(A_{i_r}^{kn}) + v_{i_r}(h_{i_{r+1}}) - v_{i_r}(h_{i_r})
\\
& \le & \alpha_{i_r} - \beta_{i_r} + v_{i_r}(h_{i_{r+1}})  - \alpha_{i_r} \\
& = &  v_{i_r}(h_{i_{r+1}}) - \beta_{i_r}  \,,
\end{IEEEeqnarray*}
where we used the induction hypothesis for the first inequality. Since this difference must be positive for $(i_r, i_{r+1})$ to be in $C$, we get that $v_{i_r}(h_{i_{r+1}}) = \alpha_{i_r}$.

Finally, we distinguish two cases, depending on whether there is another agent, besides $i_1$, who sees the good it received in this round as low-valued.
\medskip

\noindent\underline{\textbf{Case 1:} $v_{i_r}(h_{i_r}) = \alpha_{i_r}$, for all $r\in \{2, \ldots, s\}$.} In this case, we can get a new matching $\mathcal{M}'$ by  assigning each good among $h_{i_1}, \ldots, h_{i_s}$ to the `previous' agent with respect to the cycle, i.e., match $h_{i_1}$ to $i_s$, $h_{i_2}$ to $i_1$, and so on, and keeping the rest of the matching the same as $\mathcal{M}$. By the third  observation above, we have that $v_{i_r}(h_{i_{r+1}}) = \alpha_{i_r}$, for all $r\in \{2, \ldots, s\}$, whereas by our first observation about agent $i_1$'s envy, we also have $v_{i_1}(h_{i_2}) = \alpha_{i_1}$. Note that no weight is decreased going from $\mathcal{M}$ to $\mathcal{M}'$ and  $\tilde{v}_{i_1}(h_{i_2})$ is now increased to $2\big(1 + \frac{1}{2n}\big)^{n-\pi(i_1)}$ (from $\big(1 + \frac{1}{2n}\big)^{n-\pi(i_1)} $ that was in $\mathcal{M}$). This contradicts the choice of $\mathcal{M}$ as a maximum weight matching. 
\medskip

\noindent\underline{\textbf{Case 2:} $v_{i_r}(h_{i_r}) = \beta_{i_r}$, for some $r\in \{2, \ldots, s\}$.} Let $\ell$ the smallest such $r$. That is, $v_{i_x}(h_{i_x}) = \alpha_{i_x}$, for all $x\in \{2, \ldots, \ell-1\}$, but $v_{i_\ell}(h_{i_\ell}) = \beta_{i_\ell}$.
In this case, we can get a new matching $\mathcal{M}'$ by  assigning each good among $h_{i_1}, \ldots, h_{i_\ell}$ to the `previous' agent with respect to the cycle and , i.e., match $h_{i_2}$ to $i_1$, $h_{i_3}$ to $i_2$,  and so on, as well as $h_{i_1}$ to $i_\ell$, while keeping the rest of the matching the same as $\mathcal{M}$. Now we can argue like in Case 1. We have  $v_{i_r}(h_{i_{r+1}}) = \alpha_{i_r}$, for all $r\in \{2, \ldots, \ell - 1\}$, by our third  observation, whereas $v_{i_1}(h_{i_2}) = \alpha_{i_1}$ like before and, of course, $v_{i_\ell}(h_{i_1}) \ge \beta_{i_\ell}$.
Again, no weight is decreased going from $\mathcal{M}$ to $\mathcal{M}'$ and $\tilde{v}_{i_1}(h_{i_2})$ is increased from $\big(1 + \frac{1}{2n}\big)^{n-\pi(i_1)}$ to $2\big(1 + \frac{1}{2n}\big)^{n-\pi(i_1)}$ (possibly $\ell$'s weight increased as well). Like in Case 1, this contradicts the choice of $\mathcal{M}$ as a maximum weight matching. 
\medskip

We conclude that $G'$ cannot contain any cycles.
\end{proof}

Claims \ref{claim:priority_envy} and \ref{claim:priority_cycles} complete the induction. Thus,  Algorithm \ref{alg:foresight_n} builds an allocation that is \efo for every time step which is a multiple of $n$ and, because the allocation is balanced, it is also temporal-\eft.

What remains to be shown is the last part of the theorem. Suppose that at the end of some step $t_{0}$ the allocation fails to be $1/2$-\efo. Note that this happens during round $k=\lceil t_{0}/n \rceil$ and let $i, j$ be two agents, such that $v_i(A^{t_0}_i) < 0.5 v_i(A^{t_0}_j \setminus S)$ for any $S\subseteq A^{t_0}_j$ with $|S|\le 1$. It is easy to see that this can only happen if $i$ was envious of $j$ already at the end of round $k -1$; otherwise no single good added in $j$'s bundle can violate (even exact) \efo from $i$'s perspective. 
Besides this, we also claim that, if $h_i, h_j$ are the goods agents $i$ and $j$ receive, respectively, in round $k$, then $v_i(h_j) = \alpha_i$. Indeed, using property (iii) we showed above, 
even if agent $j$ receives its good first in round $k$, we  have 
\[v_i(A^{t_0}_j) - v_i(A^{t_0}_i) \le v_i(A^{(k-1)n}_j) + v_i(h_j) - v_i(A^{(k-1)n}_i) \le \alpha_i - \beta_i + v_i(h_j)\,.\]
So, if $h_j$ was low-valued for agent $i$, we would be able to eliminate $i$'s envy towards $j$ by just removing a high-valued good from  $A^{t_0}_j$ (which must exist, otherwise $v_i(A^{t_0}_j)$ and $v_i(A^{t_0}_i)$ would only differ by a single low-valued good at most).
Now, given that $v_i(h_j) = \alpha_i$ and that $i$ was already envious of $j$, it must also be the case that $v_i(h_i) = \alpha_i$, or we could repeat the exact same argument as in the proof of Claim \ref{claim:priority_envy} (i.e., construct the matching $\mathcal{M}'$ by switching  $h_i$ and $h_j$ and get a contradiction by showing it has larger weight than $\mathcal{M}$). Therefore, by the end of the round (at time step $kn$), we have $v_i(A^{kn}_i) \ge \alpha_i$. Since the allocation is temporal-\eft, at the end of any time step $t\ge kn$, we have
\[v_i(A^{t}_i)   \ge  \min_{S:|S|\le 2} v_i(A^{t}_j \setminus S)
 \ge  \min_{S:|S|\le 1} v_i(A^{t}_j \setminus S) - \alpha_i 
 \ge   \min_{S:|S|\le 1} v_i(A^{t}_j \setminus S) - v_i(A^{t}_i)\,, \]
 and, thus, $v_i(A^{t}_i)   \ge  0.5 \min_{S:|S|\le 1} v_i(A^{t}_j \setminus S)$, i.e., the allocation at the end of time step $t$ is $1/2$-\efo  
\end{proof}

Like in Section \ref{sec:foresight_2}, we can get the analog of Corollary \ref{cor:foresight_2-asymptotics} but for any number of agents.

\begin{corollary}\label{cor:foresight_n-asymptotics}  
For any $\lambda \in \mathbb{Z}_{>0}$, if $m$ is large enough, after a sufficient number of steps the allocation built by Algorithm \ref{alg:foresight_n} becomes and remains $\lambda/(\lambda+2)$-\ef, $\lambda/(\lambda+1)$-\efo, and $\lambda/(\lambda+2)$-\prop. 
\end{corollary}

\begin{proof}
The proof differs from that of Corollary \ref{cor:foresight_2-asymptotics} only in getting the guarantee for proportionality.
Assuming that by time step $t_*$ by  each agent $i$ has received value equal to at least $\lambda \alpha_i$, we have for agent $1$ \linebreak
\[\frac{1}{n}\sum_{i=1}^{n}v_1(A_i^{t}) \le \frac{1}{n} \Big[v_1(A_1^{t}) + (n-1)\Big(1 + \frac{2}{\lambda}\Big)v_1(A_1^{t})\Big] =  \Big(1 + \frac{2(n-1)}{\lambda n}\Big) v_1(A_1^{t}) \le \Big(1 + \frac{2}{\lambda}\Big) v_1(A_1^{t})\,, \]
and similarly for agents $2$ through $n$.
\end{proof}

%%%%%%%%%%%%%%%%%%%%%%%%%%%%%%%%%%%%%%%%%
%%%%%%%%%%%%%%%%%%%%%%%%%%%%%%%%%%%%%%%%%
\section{Going Beyond $2$-Value Instances}
\label{sec:beyond_2value}
%%%%%%%%%%%%%%%%%%%%%%%%%%%%%%%%%%%%%%%%%
%%%%%%%%%%%%%%%%%%%%%%%%%%%%%%%%%%%%%%%%%

As we mentioned in the Introduction, there is a simple way of approximating any additive instance via a $2$-value instance: we set a threshold for each agent and we round everything up to the maximum value this agent has for any good or down to the minimum value this agent has for any good. Of course, this naive idea does not give any guarantees, in general, but it seems like it is a very natural  
approach for personalized interval-restricted instances. Indeed, we can transfer all of our positive results to personalized interval-restricted instances at the expense of an additional multiplicative factor that is equal to the harmonic mean of the upper and lower bounds of the values an agent may have for the goods.

\begin{theorem}\label{thm:reduction}
For any personalized interval-restricted instance (augmented with foresight or not), there is a simple reduction to a personalized $2$-value instance (which is equally augmented), so that any guarantee with respect to \ef, \efo, \eft, \prop, or \mms we may obtain for the latter (e.g., via Algorithms \ref{alg:main_alg}, \ref{alg:foresight_2}, or \ref{alg:foresight_n}) can be translated to the same guarantee for the original instance at the expense of an additional multiplicative factor of $\sqrt{\alpha_i}$ for each agent $i$. 
\end{theorem}

\begin{proof}
The main idea is to use auxiliary valuation functions that aim to approximate the original valuation functions while taking only two values. Given the personalized interval-restricted valuation function $v_i$ of an agent $i$, such that $v_i(g)\in [1, \alpha_i]$, for all $g\in M$, we define the personalized $2$-value \emph{threshold}  function $\hat{v}_i$ as follows:
\[\hat{v}_{i}(g)  = \begin{cases} \alpha_i \,, & \text{if }v_i(g) > \sqrt{\alpha_i\,} \\ \sqrt{\alpha_i\,}\,, & \text{otherwise} \end{cases}\]
for any $g\in M$. It is not hard to see how $v_i$ and $\hat{v}_{i}$ are related.

\begin{claim}\label{claim:reduction}
For any set of goods $S\subseteq M$ and any agent $i\in N$, it holds that $\hat{v}_{i}(S) / \sqrt{\alpha_i\,} \le v_i(S) \le \hat{v}_{i}(S)$.
\end{claim}
\begin{proof}[Proof of Claim \ref{claim:reduction}]
\renewcommand\qedsymbol{{\scriptsize \textbf{Cl.~\ref{claim:reduction}} $\boxdot$}}
For any $g\in M$, $v_i(g)$ is rounded up in order to obtain $\hat{v}_{i}(g)$, so it is straightforward that $v_i(g) \le \hat{v}_{i}(g)$. On the other hand, $v_i(g)$ is always rounded up by a factor that is at most $\sqrt{\alpha_i\,}$, so $\sqrt{\alpha_i\,} {v}_{i}(g) \ge \hat{v}_{i}(g)$. Since both functions are additive, these inequalities extend to any set of goods.
\end{proof}

We first argue about \ef, \efo, and \eft. Suppose that at the end of some time step $t$
the allocation $(A_1^t, \ldots, A_n^t)$ is $\rho$-\ef$\!k$ (where \ef$\!0$ is just \ef) with respect to the threshold functions $\hat{v}_1, \ldots, \hat{v}_n$. Then, for any $i, j \in N$, we have
\[v_i(A^{t}_i)   \ge \frac{1}{\sqrt{\alpha_i\,}}\, \hat{v}_i(A^{t}_i) \ge  \frac{1}{\sqrt{\alpha_i\,}}\, \rho  \min_{S:|S|\le k} \hat{v}_i(A^{t}_j \setminus S)
 \ge \frac{\rho}{\sqrt{\alpha_i\,}}\, \min_{S:|S|\le k} v_i(A^{t}_j \setminus S)\,, \]
where the first and the third inequalities follow from Claim \ref{claim:reduction}. Thus, $(A_1^t, \ldots, A_n^t)$ is ${\rho}/\!{\sqrt{\alpha_i}}$-\ef$\!k$ with respect to the original functions ${v}_1, \ldots, {v}_n$.

Next, suppose that at the end of  time step $t$ the allocation $(A_1^t, \ldots, A_n^t)$ is $\rho$-\prop with respect to $\hat{v}_1, \ldots, \hat{v}_n$.
Then, similarly to the above, for any $i\in N$, we have
\[v_i(A^{t}_i)   \ge \frac{1}{\sqrt{\alpha_i\,}}\, \hat{v}_i(A^{t}_i) \ge  \frac{1}{\sqrt{\alpha_i\,}} \frac{\rho}{n} \, \hat{v}_i\Big({\textstyle \bigcup\limits_{j=1}^n} A^{t}_j \Big)
 \ge \frac{\rho}{n\sqrt{\alpha_i\,}}\, \hat{v}_i\Big({\textstyle \bigcup\limits_{j=1}^n} A^{t}_j \Big)\,, \]
where again the first and the third inequalities follow from Claim \ref{claim:reduction}. Thus, $(A_1^t, \ldots, A_n^t)$ is ${\rho}/\!{\sqrt{\alpha_i}}$-\prop with respect to the original functions ${v}_1, \ldots, {v}_n$.

We last argue about \mms. Although the idea is the same, now it is not straightforward to get the last inequality by Claim \ref{claim:reduction} in the chain of inequalities 
needed. Instead, we are going to relate the maximin shares with respect to the original and to the threshold valuation functions. For a set of goods $S\subseteq M$, let $\bmu_i^n(S)$ and $\hat{\bmu}_i^n(S)$ be the maximin shares of agent $i$ with respect to $v_i$ and to $\hat{v}_i$, respectively. 

\begin{claim}\label{claim:mms_comparison}
For any set of goods $S\subseteq M$ and any agent $i\in N$, it holds that $\bmu_i^n(S) \le \hat{\bmu}_i^n(S)$.
\end{claim}
\begin{proof}[Proof of Claim \ref{claim:mms_comparison}]
\renewcommand\qedsymbol{{\scriptsize \textbf{Cl.~\ref{claim:mms_comparison}} $\boxdot$}}
Suppose $\mathcal{T} = (T_1, \dots, T_n)$ is a maximin share defining partition of $S$ for agent $i$ with respect to $v_i$, i.e., $\bmu_i^n(S) = \min_{T_j\in \mathcal{T}} v_i(T_j)$. Now it is easy to relate the two maximin shares:
\[\hat{\bmu}_i^n(S) \ge \min_{T_j\in \mathcal{T}} \hat{v}_i(T_j) \ge \min_{T_j\in \mathcal{T}} v_i(T_j) = \bmu_i^n(S)\,,\]
where the first inequality follows by the definition of $\hat{\bmu}_i^n(S)$ (that takes the maximum over all such partitions) and the second inequality follows by Claim \ref{claim:reduction}.
\end{proof}

Suppose that at the end of some time step $t$
the allocation $(A_1^t, \ldots, A_n^t)$ is $\rho$-\mms with respect to the threshold functions $\hat{v}_1, \ldots, \hat{v}_n$. 
Then, for any $i\in N$, we have
\[v_i(A^{t}_i)   \ge \frac{1}{\sqrt{\alpha_i\,}}\, \hat{v}_i(A^{t}_i) \ge  \frac{1}{\sqrt{\alpha_i\,}}\, \rho \, \hat{\bmu}^n_i\Big({\textstyle \bigcup\limits_{j=1}^n} A^{t}_j \Big)
 \ge \frac{\rho}{\sqrt{\alpha_i\,}}\, {\bmu}^n_i\Big({\textstyle \bigcup\limits_{j=1}^n} A^{t}_j \Big)\,, \]
where, as usual the first and third inequalities follow from Claim \ref{claim:reduction}. Thus, $(A_1^t, \ldots, A_n^t)$ is ${\rho}/\!{\sqrt{\alpha_i}}$-\mms with respect to the original functions ${v}_1, \ldots, {v}_n$.
\end{proof}

For any personalized interval-restricted instance let $\alpha_* := \max_{i\in N}\sqrt{\alpha_i\,}$. Then Theorem \ref{thm:reduction}, combined with Corollary \ref{cor:type1} or Theorem \ref{thm:foresight_n-b}, directly implies the following corollaries.

\begin{corollary}\label{cor:interval-main}
For any personalized interval-restricted instance, we can construct a ${1}/{\sqrt{a_*}(2n-1)}$-temporal-\mms allocation. Moreover, for any  agent $i$ who sees at least $n$ high-valued goods with respect to $\hat{v}_i$ of Theorem \ref{thm:reduction}, this guarantee eventually improves to $\Omega(1/\sqrt{a_i\,})$.
\end{corollary}

\begin{corollary}\label{cor:interval-foresight}
For any personalized interval-restricted instance augmented with foresight of length $n-1$, we can construct a ${1}/{\sqrt{a_*\,}}$-temporal-\eft allocation that is also ${1}/{\sqrt{a_*\,}}$-\efo (and, thus, ${1}/{n\sqrt{a_*\,}}$-\mms) for every time step $t = kn, k\in \mathbb{Z}_{\ge 0}$. If at any step $t_{0}$ the allocation fails to be ${1}/{2\sqrt{a_*\,}}$-\efo, then it remains ${1}/{2\sqrt{a_*\,}}$-\efo  at the end of every time step $t\ge \lceil t_0 / n \rceil n$. 
\end{corollary}

%%%%%%%%%%%%%%%%%%%%%%%%%%%%%%%%%%%%%%%%%
%%%%%%%%%%%%%%%%%%%%%%%%%%%%%%%%%%%%%%%%%
\section{Discussion and Open Questions}
\label{sec:discussion}
%%%%%%%%%%%%%%%%%%%%%%%%%%%%%%%%%%%%%%%%%
%%%%%%%%%%%%%%%%%%%%%%%%%%%%%%%%%%%%%%%%%
In this paper we study a prior-free online fair division setting where the items are indivisible goods and we focus on the design of deterministic algorithms with solid worst-case fairness guarantees that hold frequently, if not at every time step. By restricting the input space to personalized $2$-value (or interval-restricted) instances we are able to obtain nontrivial guarantees that are not possible for the general additive case. We see this as a main take-home message of this work; despite the existence of strong impossibility results, there are meaningful restrictions which can lead to technically interesting findings, broadening our understanding of online fair division. So, the most natural direction for future work is to identify such restrictions and push the boundaries of positive results accordingly.

Another promising direction is to fully explore the power of knowing the future. Our work does not answer whether it is possible to efficiently utilize  foresight which is \emph{sublinear}  in $n$ for personalized $2$-value instances. In fact, Algorithms \ref{alg:foresight_2} and \ref{alg:foresight_n} only use the (linear) information they have every $n$ steps and ignore it otherwise. We suspect that it is possible to design algorithms that build temporal-\efo allocations with linear foresight in our setting, although there are simple examples suggesting that this cannot be done via balanced allocations (like the ones our algorithms construct).

Finally, given that personalized interval-restricted instances are already very expressive, it would be particularly interesting to get tight results directly for those. Although  the dependency on $\alpha_*$ cannot be completely removed (since for large enough $\alpha_*$ the impossibility results of \citet{HePPZ19} and \citet{ZhouBW23} can be replicated), it is likely that the guarantees of Corollaries \ref{cor:interval-main} and \ref{cor:interval-foresight} are not tight.

%%%%%%%%%%%%%%%%%%%%%%%%%%%%%%%%%%%%%%%%%
%%%%%%%%%%%%%%%%%%%%%%%%%%%%%%%%%%%%%%%%%
\section*{Acknowledgments}
%%%%%%%%%%%%%%%%%%%%%%%%%%%%%%%%%%%%%%%%%
%%%%%%%%%%%%%%%%%%%%%%%%%%%%%%%%%%%%%%%%%
This work was partially supported by the project MIS 5154714 of the National Recovery and Resilience Plan Greece $2.0$ funded by the European Union under the NextGenerationEU Program.

This work was partially supported by the framework of the H.F.R.I call “Basic research Financing (Horizontal support of all Sciences)” under the National Recovery and Resilience Plan ``Greece 2.0'' funded by the European Union – NextGenerationEU (H.F.R.I. Project Number: 15877).

This work was partially supported by the NWO Veni project No.~VI.Veni.192.153.

% \newpage

\bibliographystyle{plainnat}
\bibliography{fairdiv_refs}

\end{document}